\documentclass[a4paper,11pt]{article}
\usepackage{amsmath,amsfonts,amssymb,amsthm,latexsym}
\usepackage{amscd}
\usepackage{mathtools}
\numberwithin{equation}{section}
\usepackage[latin1]{inputenc}
\usepackage{lmodern}
\usepackage[T1]{fontenc}
\usepackage{verbatim}
\usepackage{enumerate}
\usepackage{xcolor,graphicx}
\usepackage{tikz}
\usetikzlibrary{shapes}
\newtheorem{theorem}{Theorem}[section]
\newtheorem{lemma}[theorem]{Lemma}
\newtheorem{proposition}[theorem]{Proposition}
\newtheorem{corollary}[theorem]{Corollary}

\newtheorem{assumption}{Assumption}
\theoremstyle{definition}
\newtheorem*{definition}{Definition}
\theoremstyle{remark}
\newtheorem{remark}{Remark}
\newtheorem{example}{Example}
\definecolor{labelkey}{rgb}{0,0.08,0.45}
\definecolor{refkey}{rgb}{0,0.6,0.0}
\definecolor{Brown}{rgb}{0.45,0.0,0.05}
\definecolor{dgreen}{rgb}{0.00,0.49,0.00}
\definecolor{dblue}{rgb}{0,0.08,0.75}
\usepackage{ifpdf}
\usepackage[pdftex,colorlinks,bookmarks=true,unicode=true,pdftitle={coalitions in nonatomic network congestion games},pdfauthor={Cheng WAN}]{hyperref}
\hypersetup{
    linktocpage=true,
    citecolor=dblue,    
    linkcolor=dgreen    
}

\title{Coalitions in nonatomic network congestion games}
\author{Cheng Wan\footnote{{\em Email}: cheng.wan.2005@polytechnique.org}\\ \small Combinatoire et Optimisation, Institut de Math\'emathiques de Jussieu (CNRS, UMR 7586)\\  \small Facult\'e de Math\'ematiques, Universit\'e P. et M. Curie (Paris~6), 4 place Jussieu, 75005 Paris, France}
\date{May 12, 2012}

\begin{document}
\maketitle

\begin{abstract}
This work shows that the formation of a finite number of coalitions in a nonatomic network congestion game benefits everyone. At the equilibrium of the composite game played by coalitions and individuals, the average cost to each coalition and the individuals' common cost are all lower than in the corresponding nonatomic game (without coalitions). The individuals' cost is lower than the average cost to any coalition. Similarly, the average cost to a coalition is lower than that to any larger coalition. Whenever some members of a coalition become individuals, the individuals' payoff is increased. In the case of a unique coalition, both the average cost to the coalition and the individuals' cost are decreasing with respect to the size of the coalition. In a sequence of composite games, if a finite number of coalitions are fixed, while the size of the remaining coalitions goes to zero, the equilibria of these games converge to the equilibrium of a composite game played by the same fixed coalitions and the remaining individuals.
\end{abstract}


\section{Introduction}\label{sect1}
This paper considers the impact of introducing coalitions in network congestion games played by nonatomic individuals, namely nonatomic routing games. These games belong to a more general class of noncooperative games played by a continuum of anonymous identical players, each of whom has a negligible effect on the others.

First, let us cite some historic references on routing games, in particular, on coalitions in such games.

Beckman, McGuire and Winston \cite{Bec53} first formulated Wardrop equilibrium (Wardrop \cite{War52}) in nonatomic congestion games as an optimal solution of a convex programming problem, and thus proved its existence under weak conditions on the cost functions.

A coalition of nonatomic individuals of total weight $T$ behaves the same way as an atomic player who holds a flow of weight $T$ that can be split and sent by different paths. Routing games with finitely many atomic players holding splittable flow (called atomic splittable games) were first examined by Haurie and Marcotte \cite{Hau85}. They focused on the asymptotic behavior of Nash equilibria in such games. By characterizing a Nash equilibrium in an atomic splittable game and a Wardrop equilibrium in the corresponding nonatomic game by two variational inequalities, they proved that the former converges to the latter, when the number of atomic players tends to infinity. This result will be extended in this paper.

Harker \cite{Ha88} first studied {\em composite games} (that he called mixed games), where atomic players holding splittable flow (or coalitions) and nonatomic individuals play together. He characterized a composite-type equilibrium by a variational inequality, and thus proved the existence of a solution under some weak conditions on the cost functions as well as its uniqueness under more stringent conditions.

Orda, Rom and Shimkin \cite{Orda93} made a detailed study on the uniqueness and other properties of Nash equilibria in atomic splittable games on {\em two-terminal parallel-link} networks. This specific setting will be adopted in this paper, where their results will be extended. Richman and Shimkin \cite{Rich07} extended their results to composite congestion games in nearly parallel-link networks.

For the impact of coalitions on the equilibrium costs, Cominetti, Correa and Stier-Moses \cite{Com09} showed that, in the atomic splittable case where the atomic players are identical, the social cost at the equilibrium of the game is bounded by that of the corresponding nonatomic game, under weak conditions on the cost functions.

Hayrapetyan, Tardos and Wexler \cite{Hay06} proved that the formation of coalitions (that they called collusion) reduces the social cost in a two-terminal parallel-link network. Although stronger conditions on the cost functions are needed in this paper, our results prove that the formation of coalitions benefits everyone.

Apart from the consequence of the formation of coalitions on the equilibrium costs, this paper also studies how this impact varies with the structure of coalitions.

\subsection{A Sketch of the model}
A continuum of nonatomic individuals are commuters in a two-terminal parallel-arc (directed) network. Their common origin and common destination are the only two vertices, which are connected by a finite set of parallel arcs. The per-unit traffic cost of an arc depends only on the total weight of the flow on it. A pure strategy of an individual is an arc by which she goes from the origin to the destination. Nash equilibria in such nonatomic games are usually called {\em Wardrop equilibria} (WE for short)~\cite{War52}. At a WE, the arc chosen by an individual costs no more than any other available arc, hence it has the lowest cost in the network. The individuals have the same cost at a WE.

A composite routing game is played by a finite number of disjoint coalitions formed by some of the individuals and the remaining individuals. A coalition is specified by its size. Within a coalition, a coordinator assigns an arc to each member, with the objective of minimizing their total cost. An equilibrium in this game is called {\em composite equilibrium} (CE for short), since it is Nash-type for the coalitions and Wardrop-type for the individuals. All the individuals have the same cost at a CE, while the average costs to the coalitions may differ.

\subsection{Main results}
After recalling the {\em existence} and the {\em uniqueness} of the CE of a composite game under certain conditions on the cost functions, five main results are obtained:
\begin{enumerate}
  \item At the CE, the average social cost, the individuals' cost and the average cost to each coalition are lower than the equilibrium cost at the WE of the corresponding nonatomic game.
  \item At the CE, the average cost to a coalition is lower than that to any other larger coalition. If a coalition sends flow on a certain arc, then any other larger coalition sends more on it.
  \item If some members quit a coalition to become individuals, the individuals' cost is increased at the corresponding CE.
  \item If there is only one coalition, the social cost, the average cost to the coalition, and the individuals' cost at the CE are all decreasing with respect to the size of the unique coalition.
  \item If, in a sequence of composite games, a finite number of coalitions are fixed, and the maximum size of the remaining coalitions tends to zero, the sequence of equilibrium of these games converges to the equilibrium of a game played by the same fixed coalitions and the remaining individuals.
\end{enumerate}

\subsection{Organization of the work}
The paper is organized in the following way. Section~\ref{sect2} provides a detailed description of the model as well as characterizations of the CE in different formulations. The existence and the uniqueness of the CE will be recalled. Section~\ref{sect3} analyzes some important properties of the CE. Section~\ref{sect4} deals with the impact of the formation of coalitions by comparing the players' costs at the WE of the corresponding nonatomic game and those at the CE. Section~\ref{sect5} considers the impact of the composition of the players on the CE costs: first, how the equilibrium costs vary with the size of a unique coalition; second, how the individuals' cost varies when some members of a coalition become individuals. Section~\ref{sect6} focuses on the asymptotic behavior of CE, by fixing some coalitions while letting the remaining coalitions vanish. Section~\ref{sect7} discusses some problems for future research.

\section{The model and characterization of an equilibrium}\label{sect2}
\subsection{Model and notations}
\paragraph{Network and arc costs $(\mathcal{R},\,\mathbf{c})$.}
Let the set of identical anonymous nonatomic individuals be described by the unit real interval $I=[\,0,\,1]$, endowed with the Lebesgue measure $\mu$. The players' common origin is vertex $O$, and their common destination is vertex $D$. The finite set of parallel arcs between $O$ and $D$ is denoted by $\mathcal{R}$, with $R=|\mathcal{R}|$ its cardinality. Let $\mathbf{c}=(c_{r})_{r\in \mathcal{R}}$ be the vector of the per-unit arc cost functions: for every arc $r$, $x \mapsto c_{r}(x)$ is a real function defined on a neighborhood $U$ of $[\,0,\,1]$. The per-unit cost of an arc only depends on the total weight of the flow on it. The network is characterized by the pair $(\mathcal{R},\,\mathbf{c})$.

The following assumption is made {\em throughout} this paper.
\begin{assumption}\label{cost_assumption_1}
For every arc $r$ in $\mathcal{R}$, the cost function $c_{r}$ is strictly increasing, convex and continuously differentiable on $U$, and nonnegative on $[\,0,\,1]$.
\end{assumption}
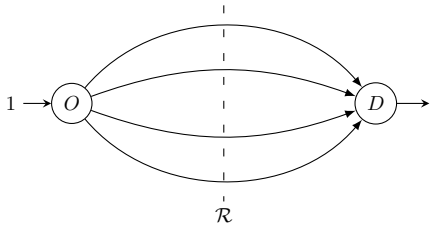
\begin{figure}[!htbp]
\begin{center}
\begin{tikzpicture}
 \node[draw,circle,scale=0.7] (O)at(-2,0) {$O$};
 \node[draw,circle,scale=0.7] (D)at(2,0) {$D$};
 \node (P)[scale=0.7] at (-2.8,0) {$1$};
\node (Q)[scale=0.7] at (2.8,0) {};
\draw [->,>=stealth] (P) to (O);
\draw [->,>=stealth] (D) to (Q);
\draw[->,>=latex] (O) to[bend left=50] (D);
\draw[->,>=latex] (O) to[bend left=20] (D);
\draw[->,>=latex] (O) to[bend right=20] (D);
\draw[->,>=latex] (O) to[bend right=50] (D);
\draw [loosely dashed] (0,1.3) -- (0,-1.3)node[below,scale=0.7]{$\mathcal{R}$};
\end{tikzpicture}
\end{center}
\end{figure}
\paragraph{Composite routing game $\Gamma(\mathcal{R},\mathbf{c},\mathbf{T})$.}
Suppose that $K$ coalitions are formed in the set of individuals $I$, with $K\in \mathbb{N}=\{0,1,2,\cdots\}$. The family of coalitions is denoted by $\mathcal{K}=\{1,\,\ldots,\,K\}$. Every coalition behaves like an atomic player holding a splittable flow. The remaining individuals are independent nonatomic players. For a coalition $k\in \mathcal{K}$, the measurable set of its members is denoted by $I^{k}$, a subset of $I$, and its total weight is denoted by $T^{k}=\mu(I^{k})$. Let $I^{0}$ denote the set of individuals so that $I^{0}=[\,0,\,1]\backslash \cup_{k\in \mathcal{K}}I^{k}$, and its weight is $T^{0}=\mu(I^{0})=1-\sum_{k\in \mathcal{K}}T^{k}$. Without loss of generality, it is assumed that $T^{1}\geq T^{2}\geq \cdots \geq T^{K}$. Let us define $\mathbf{T}=(\,T^{0};\,T^{1},\ldots,\,T^{K}\,)$. Let $\Gamma(\mathcal{R},\mathbf{c},\mathbf{T})$ be the composite routing game played by these $K$ coalitions and the remaining individuals in the network $(\mathcal{R},\,\mathbf{c})$.

Two particular cases should be mentioned. First, if $I^{0}=[\,0,\,1]$ and $K=0$, there is no coalition so that the game is a nonatomic one, denoted simply by $\Gamma(\mathcal{R},\mathbf{c})$, and the equilibria there are WE. Second, if $I^{0}$ is empty, i.e. $T^{0}=0$, the game is an atomic splittable one with $K$ atomic players, and the equilibria there are {\em Nash equilibria} (NE for short) in its usual sense; in particular, if $K=1$, i.e, there is a global coalition, the equilibrium is obtained by solving the optimization problem of searching for the social optimum.

\paragraph{Strategies and flow configurations.}
In the game $\Gamma(\mathcal{R},\mathbf{c},\mathbf{T})$, as the individuals are identical and anonymous, only the total weight sent by each coalition on each arc counts. The strategy profile of the individuals ({\em resp.} the strategy of coalition $k$) is specified by the {\em flow configuration} (flow for short) $\mathbf{x}^{0}$ ({\em resp.} $\mathbf{x}^{k}$) defined by
\begin{equation*}
\mathbf{x}^{0}=(x^{0}_{r})_{r\in \mathcal{R}}\quad
(resp. \quad \mathbf{x}^{k}=(x^{k}_{r})_{r\in \mathcal{R}}),
\end{equation*}
where $x^{0}_{r}$ ({\em resp.} $x^{k}_{r}$) is the total weight of the individuals ({\em resp.} of coalition $k$) on arc $r$.

A strategy profile is specified by $\mathbf{x}=(\mathbf{x}^{0},\,\mathbf{x}^{1},\,\ldots,\,\mathbf{x}^{K})$, a point in $\mathbb{R}^{(1+K)\times R}$.

The {\em feasible flow set} of the individuals ({\em resp.} of coalition $k$) is a convex compact subset of $ \mathbb{R}^{R}$, defined by
\begin{align*}
F^{0} & =\bigl\{\mathbf{x}^{0}\in \mathbb{R}^{R}\;|\; \forall \,r\in \mathcal{R},\,x^{0}_{r}\geq 0\,; \,\sum_{r\in \mathcal{R}}x^{0}_{r}=T^{0}\bigr\},\\
\bigl(\text{{\em resp.}} \quad F^{k}
 &= \bigl\{\mathbf{x}^{k}\in \mathbb{R}^{R}\;|\; \forall \,r\in \mathcal{R},\,x^{k}_{r}\geq 0\,; \,\sum_{r\in \mathcal{R}}x^{k}_{r}=T^{k}\bigr\}\bigr).
\end{align*}

The {\em feasible flow set} $F$ of the game $\Gamma(\mathcal{R},\mathbf{c},\mathbf{T})$ is a convex compact subset of $ \mathbb{R}^{(1+K)\times R}$, defined by $F=F^{0} \times F^{1} \times \cdots \times F^{K}$.

The {\em aggregate flow} $\mathbf{x}'$ induced by $\mathbf{x}$ is a vector in $\mathbb{R}^{R}$, defined by $\mathbf{x}'=(x_{r})_{r\in \mathcal{R}}$, where $x_{r}=x^{0}_{r}+\sum_{k\in \mathcal{K}}x^{k}_{r}$ is the aggregate weight on arc $r$.

For coalition $k$, the vector $\mathbf{x}^{-k}$ is a point in $F^{-k}= \prod_{l\in \{0\}\cup \mathcal{K}\setminus \{k\}}F^{l}$, defined by $\mathbf{x}^{-k}=(\mathbf{x}^{l})_{l\in \{0\}\cup \mathcal{K}\setminus \{k\}}$. For all arc $r$, define $x^{-k}_{r}=x^{0}_{r}+\sum_{l\in \mathcal{K}\setminus \{k\}}x^{l}_{r}$.

\paragraph{Average costs and marginal costs.}
The {\em average cost} to the individuals, the {\em average cost} to coalition $k$ and the {\em average social cost} are respectively defined by
\[Y^{0}(\mathbf{x}) = \frac{1}{T^{0}} \sum_{r\in \mathcal{R}}x^{0}_{r}c_{r}(x_{r}),\;\;Y^{k}(\mathbf{x})= \frac{1}{T^{k}}\sum_{r\in \mathcal{R}}x^{k}_{r}c_{r}(x_{r}),\;\;Y(\mathbf{x})= \sum_{r\in \mathcal{R}}x_{r}c_{r}(x_{r}).\]
As the total weight of the players is normalized to 1, the average social cost is just the social cost.


The total cost to coalition $k$ is denoted by $u^{k}(\mathbf{x})=T^{k}\cdot Y^{k}(\mathbf{x})=\sum_{r\in \mathcal{R}} x^{k}_{r}c_{r}(x_{r})$.

Following Harker~\cite{Ha88}, the {\em marginal cost} function of coalition $k$ is defined by
\begin{equation*}
 \hat{c}^{k}(\mathbf{x})=(\hat{c}^{k}_{r}(\mathbf{x}))_{r\in \mathcal{R}},\;\,\text{where } \, \hat{c}_{r}^{k}(\mathbf{x})=c_{r}(x_{r})+x^{k}_{r}c'_{r}(x_{r}).
\end{equation*}

Notice that $\hat{c}^{k}(\mathbf{x})$ is the gradient of $u^{k}(\mathbf{x})$ with respect to $\mathbf{x}^{k}$. More precisely,
\begin{equation*}
\hat{c}^{k}(\mathbf{x})=\nabla_{\mathbf{x}^{k}}\,  u^{k}(\mathbf{x}^{k},\,\mathbf{x}^{-k})=\Bigl( \,\frac{\partial u^{k}}{\partial x^{k}_{r}}\bigl(\mathbf{x}\bigr) \Bigr)_{r\in \mathcal{R}}\,.
\end{equation*}

\subsection{Characterizing equilibria: existence and uniqueness}
The following definition of a CE (Harker \cite{Ha88}) consists of two parts: the first for the individuals and the second for the coalitions.
\begin{definition}[Composite equilibrium]
A point $\mathbf{x}^{*}=(\,\mathbf{x}^{*0},\,\mathbf{x}^{*1},\,\ldots,\,\mathbf{x}^{*K}\,)$ in $F$ is a {\em CE} of the game $\Gamma(\mathcal{R},\mathbf{c},\mathbf{T})$ if
\begin{align}
&\forall\, r\in \mathcal{R},\quad \text{ if }\;x^{*0}_{r}>0, \;\text{ then } \;r\,\in\, \arg\min_{s\in \mathcal{R}} \, c_{s}\left(x^{*}_{s}\right);\label{mixed-cond-indiv}\\
&\forall\, k\in \mathcal{K},\quad \mathbf{x}^{*k} \, \text{ minimizes }\,u^{k}(\mathbf{x}^{k},\,\mathbf{x}^{*-k}) \;\text{ on } F^{k}.\label{mixed-cond-coal}
\end{align}
\end{definition}
\begin{proposition} [Characterization of a CE] \label{eq_characterize}
The following are equivalent.
\begin{enumerate}
  \item $\mathbf{x}^{*}=(\,\mathbf{x}^{*0},\,\mathbf{x}^{*1},\,\ldots,\,\mathbf{x}^{*K}\,)$ in $F$ is a CE.
  \item (marginal cost formulation) \,$\mathbf{x}^{*}=(\,\mathbf{x}^{*0},\,\mathbf{x}^{*1},\,\ldots,\,\mathbf{x}^{*K}\,)$ in $F$ satisfies
\begin{align}
\label{cost_nonatom}
\forall\, r \in \mathcal{R},\quad  \text{ if }\; x^{*0}_{r}> 0, \quad\text{ then }\; \forall s\in\mathcal{R},\; c_{r}(x^{*}_{r}) &\leq c_{s}(x^{*}_{s});\\
\nonumber
\forall\, k \in \mathcal{K},\quad \text{ if }\; x^{*k}_{r}> 0, \quad \text{ then }\; \forall s\in\mathcal{R},\;\hat{c}_{r}^{k}(x^{*}_{r}) &\leq \hat{c}_{s}^{k}(x^{*}_{s}),\\
\label{cost_atom}
\text{i.e. }\qquad c_{r}(x^{*}_{r})+x^{*k}_{r}c'_{r}(x^{*}_{r}) &\leq c_{s}(x^{*}_{s})+x^{*k}_{s}c'_{s}(x^{*}_{s}).
\end{align}
  \item (variational inequality formulation) \, $\mathbf{x}^{*}=(\,\mathbf{x}^{*0},\,\mathbf{x}^{*1},\,\ldots,\,\mathbf{x}^{*K}\,)$ in $F$ satisfies
\begin{equation}\label{VIP_mix}
\langle\,\mathbf{c}(\mathbf{x}^{*}),\,\mathbf{x}^{0}-\mathbf{x}^{*0}\,\rangle+\sum_{k\in \mathcal{K}}\langle\,\hat{\mathbf{c}}^{k}(\mathbf{x}^{*}),\,\mathbf{x}^{k}-\mathbf{x}^{*k}\,\rangle\geq 0, \quad \forall\; \mathbf{x}= (\,\mathbf{x}^{0},\,\mathbf{x}^{1},\,\ldots,\,\mathbf{x}^{K}\,) \in F.
\end{equation}
\end{enumerate}
Here, $\langle\, \cdot,\cdot \rangle$ stands for the standard inner product operator on the Euclidean spaces.
\end{proposition}
\begin{proof}
(i) $\Leftrightarrow$ (ii): For the individuals, \eqref{cost_nonatom} is simply a reformulation of  \eqref{mixed-cond-indiv}. For the coalitions, in order to show that \eqref{mixed-cond-coal} is equivalent to \eqref{cost_atom}, let us first prove that for coalition $k$, $u^{k}(\mathbf{x}^{k},\,\mathbf{x}^{-k})$ is convex in $\mathbf{x}^{k}$ for any given $\mathbf{x}^{-k}$ in $F^{-k}$.

Indeed, for any $r$ in $\mathcal{R}$, since $c_{r}$ is convex and strictly increasing,
\begin{align*}
c_{r}(y^{k}_{r}+x^{-k}_{r})\,&\geq \, c_{r}(x^{k}_{r}+x^{-k}_{r})+(y^{k}_{r}-x^{k}_{r})\,c'_{r}(x^{k}_{r}+x^{-k}_{r})\\
\Rightarrow \qquad y^{k}_{r}\,c_{r}(y^{k}_{r}+x^{-k}_{r})\, &\geq \, y^{k}_{r}\,c_{r}(x_{r})+y^{k}_{r}\,(y^{k}_{r}-x^{k}_{r})\,c'_{r}(x_{r}) \geq  y^{k}_{r}\,c_{r}(x_{r})+x^{k}_{r}\,(y^{k}_{r}-x^{k}_{r})\,c'_{r}(x_{r})\\
& = \,x^{k}_{r}\,c_{r}(x_{r})+(y^{k}_{r}-x^{k}_{r})\,\left[ c_{r}(x_{r})+x^{k}_{r}\,c'_{r}(x_{r})\right].\\
\Rightarrow \,\sum_{r\in \mathcal{R}}\,y^{k}_{r}\,c_{r}(y^{k}_{r}+x^{-k}_{r})\, &\geq \, \sum_{r\in \mathcal{R}}\,x^{k}_{r}\,c_{r}(x_{r})+\sum_{r\in \mathcal{R}}\,(y^{k}_{r}-x^{k}_{r})\,\left[ c_{r}(x_{r})+x^{k}_{r}\,c'_{r}(x_{r})\right],
\end{align*}
which implies that
\begin{equation}\label{u_convex}
 u^{k}(\mathbf{y}^{k},\,\mathbf{x}^{-k})\geq u^{k}(\mathbf{x}^{k},\,\mathbf{x}^{-k}) +
\langle \, \nabla_{\mathbf{x}^{k}}\,  u^{k}(\mathbf{x}^{k},\,\mathbf{x}^{-k}),\,\mathbf{y}^{k}-\mathbf{x}^{k}\,\rangle,\quad \forall \, \,\mathbf{x}^{k}, \,\mathbf{y}^{k}\,\in F^{k}.
\end{equation}

Thus, $\mathbf{x}^{*k}$ minimizes the convex function $u^{k}(\mathbf{x}^{k},\,\mathbf{x}^{*-k})$ on the convex compact set $F^{k}$ if, and only if, $\langle\,\nabla_{\mathbf{x}^{k}}\, u^{k}(\mathbf{x}^{*}),\,\mathbf{x}^{k}-\mathbf{x}^{*k}\,\rangle\geq 0$ for all $\mathbf{x}^{k}\,\in F^{k}$ or, equivalently,
\begin{equation}\label{VIP_coal}
 \langle\,\hat{\mathbf{c}}^{k}(\mathbf{x}^{*}),\,\mathbf{x}^{k}-\mathbf{x}^{*k}\,\rangle\geq 0,\quad \forall \; \mathbf{x}^{k}\,\in F^{k}.
\end{equation}

Let us set $\hat{c}^{k}=\min_{r\in \mathcal{R}}\hat{c}_{r}^{k}(\mathbf{x}^{*})$. Then, $\sum_{r\in \mathcal{R}} (\,\hat{c}^{k}_{r}(\mathbf{x}^{*})-\hat{c}^{k})\,(x^{k}_{r}-x^{*k}_{r})=\sum_{r\in \mathcal{R}} \,\hat{c}^{k}_{r}(\mathbf{x}^{*}) \,(x^{k}_{r}-x^{*k}_{r})-  \hat{c}^{k}\sum_{r\in \mathcal{R}} \,(x^{k}_{r}-x^{*k}_{r})=\sum_{r\in \mathcal{R}} \,\hat{c}^{k}_{r}(\mathbf{x}^{*}) \,(x^{k}_{r}-x^{*k}_{r})-\hat{c}^{k} (T^{k}-T^{k})=\langle\,\hat{\mathbf{c}}^{k}(\mathbf{x}^{*}),\,\mathbf{x}^{k}-\mathbf{x}^{*k}\,\rangle$. Consequently, \eqref{VIP_coal} is equivalent to
\begin{equation}\label{VIP_coal_bis}
\sum_{r\in \mathcal{R}} (\,\hat{c}^{k}_{r}(\mathbf{x}^{*})-\hat{c}^{k})\,(x^{k}_{r}-x^{*k}_{r})\geq 0.
\end{equation}

It remains to show that \eqref{VIP_coal_bis} is equivalent to \eqref{cost_atom}.

\eqref{cost_atom} $\Rightarrow$ \eqref{VIP_coal_bis}: According to \eqref{cost_atom},
\begin{equation*}
(\,\hat{c}^{k}_{r}(\mathbf{x}^{*})-\hat{c}^{k})\,(x^{k}_{r}-x^{*k}_{r})=
\begin{cases}
\,(\,\hat{c}^{k}_{r}(\mathbf{x}^{*})-\hat{c}^{k})\,x^{k}_{r}\geq 0,\;& \text{ if }  x^{*k}_{r}=0,\\
\,0,& \text{ if } x^{*k}_{r}>0.
\end{cases}
\end{equation*}

Thus, $\sum_{r\in \mathcal{R}} (\,\hat{c}^{k}_{r}(\mathbf{x}^{*})-\hat{c}^{k})\,(x^{k}_{r}-x^{*k}_{r})\geq 0$.

\eqref{VIP_coal_bis} $\Rightarrow$ \eqref{cost_atom}: Let us define an auxiliary flow $\mathbf{x}^{k}$ in $F^{k}$ as follows: ${x}^{k}_{r}=0$ if $\hat{c}^{k}_{r}(\mathbf{x}^{*})>\hat{c}^{k}$, and ${x}^{k}_{r}=\frac{T^{k}}{m}$ if $\hat{c}^{k}_{r}(\mathbf{x}^{*})=\hat{c}^{k}$. Here $m=|\{r\in \mathcal{R}\,|\, \hat{c}^{k}_{r}(\mathbf{x}^{*})=\hat{c}^{k}\}|$, the number of arcs whose marginal cost to coalition $k$ at $\mathbf{x}^{*}$ are the smallest in the network. Then, for this specific $\mathbf{x}^{k}$, \eqref{VIP_coal_bis} implies that $\sum_{r\in \mathcal{R}, \hat{c}^{k}_{r}(\mathbf{x}^{*})>\hat{c}^{k}} (\,\hat{c}^{k}_{r}(\mathbf{x}^{*}\,)-\hat{c}^{k})\,(-x^{*k}_{r})\geq 0$. Consequently, $x^{*k}_{r}=0$ if $\hat{c}^{k}_{r}(\mathbf{x}^{*})>\hat{c}^{k}$, which leads to \eqref{cost_atom}.

(ii) $\Leftrightarrow$ (iii): By the same argument used above for the equivalence between \eqref{cost_atom} and \eqref{VIP_coal}, one can show that \eqref{cost_nonatom} is equivalent to
\begin{equation}\label{VIP_indiv}
\left\langle\,\mathbf{c}(\mathbf{x}^{*}),\,\mathbf{x}^{0}-\mathbf{x}^{*0}\,\right\rangle\geq 0, \quad \forall\; \mathbf{x}^{0}\in F^{0}.
\end{equation}

The variational inequalities \eqref{VIP_coal} and \eqref{VIP_indiv} imply immediately \eqref{VIP_mix}. For the converse, it is enough to take an $\mathbf{x}= (\,\mathbf{x}^{0},\,\mathbf{x}^{1},\,\ldots,\,\mathbf{x}^{K}\,)$ in $F$ such that $\mathbf{x}^{l}=\mathbf{x}^{*l}$ for all $l$ in $\mathcal{K}$ ({\em resp.} $\mathbf{x}^{l}=\mathbf{x}^{*l}$ for all $l$ in  $\{0\}\cup \mathcal{K}\setminus\{k\}$) to get \eqref{VIP_indiv} ({\em resp.} \eqref{VIP_coal}).

Thus, one has shown that \eqref{cost_nonatom} and \eqref{cost_atom} are equivalent to \eqref{VIP_mix}.
\end{proof}
\begin{remark} \rm
(iii) has been proven for the specific cases of NE and WE as well as for CE: a WE was characterized as the solution of a variational inequality problem by Smith~\cite{Smi79} and Dafermos~\cite{Daf80}, and as the solution of a nonlinear complementarity problem by Aashtiani and Magnanti~\cite{Aas81}. Variational inequalities were used to characterize a NE in atomic splittable games by Haurie and Marcotte~\cite{Hau85}, and a CE in composite games by Harker~\cite{Ha88}.

Condition \eqref{cost_atom} shows that the marginal costs $(\hat{c}_{r}^{k})_{r\in \mathcal{R}}$ play the same role for coalition $k$ as $(c_{r})_{r\in \mathcal{R}}$ for the individuals: at the CE, all the arcs used by coalition $k$ have the lowest marginal cost and, {\em a fortiori}, the same one. For flow $\mathbf{x}\in F$, $\hat{c}_{r}^{k}(\mathbf{x})=c_{r}(x_{r})+x^{k}_{r}c'_{r}(x_{r})$ is a function of only two variables $x_{r}^{k}$ and $x_{r}$. Besides, according to Assumption~\ref{cost_assumption_1}, it is strictly increasing in both of them.
\end{remark}

\begin{theorem}[Existence and uniqueness of CE]\label{unique_mix}
In a composite game, a CE exists, and it is unique.
\end{theorem}
\begin{proof}
The variational inequality formulation for CE \eqref{VIP_mix} is used to prove its existence. Theorem 3.1 in Kinderlehrer and Stampacchia~\cite[p.12]{Kin86} states that the variational inequality problem \eqref{VIP_mix} admits a solution if $F$ is a convex compact set, and if $\hat{\mathbf{c}}^{k}$ and $\mathbf{c}$ are continuous. According to Assumption~\ref{cost_assumption_1}, these conditions are satisfied.

For the uniqueness of CE, see Richman and Shimkin \cite[Theorem~4.1]{Rich07}.
\end{proof}
\begin{remark} \rm
For the nonatomic routing game $\Gamma(\mathcal{R},\mathbf{c})$, a WE exists if the cost functions $c_{r}$'s are continuous. If they are furthermore strictly increasing on $U$, then the WE is unique. See Patriksson~\cite[Theorems 2.4, 2.5]{Pat94} for a proof.
\end{remark}
\section{A detailed study on CE}\label{sect3}
Let us consider a composite game $\Gamma(\mathcal{R},\mathbf{c},\mathbf{T})$. This section focuses on the properties of its unique CE, denoted by $\mathbf{x}$ here and in Section~\ref{sect4}.

First, some notations are recalled or given.

$\mathcal{R}^{0}(\mathbf{x})=\{ r \in \mathcal{R}\,|\, x^{0}_{r}>0\}\subset \mathcal{R}$ is the support of $\mathbf{x}^{0}$.

$\mathcal{R}^{k}(\mathbf{x})=\{ r \in \mathcal{R}\,|\, x^{k}_{r}>0\}\subset \mathcal{R}$ is the support of $\mathbf{x}^{k}$, for coalition $k$.

$c^{0}(\mathbf{x})$ is the lowest arc cost in the network.

$\hat{c}^{k}(\mathbf{x})$ is the marginal cost to coalition $k$ of every arc used by it.

$Y^{0}(\mathbf{x})$ is the common cost to all the individuals. $Y^{0}(\mathbf{x})=c^{0}(\mathbf{x})$.

$Y^{k}(\mathbf{x})$ is the average cost to coalition $k$.

$\underline{Y}^{k}(\mathbf{x})=\min_{r \in \mathcal{R}^{k}} c_{r}(x_{r})$ is the lowest arc cost of the arcs used by coalition $k$.

$Y(\mathbf{x})$ is the social cost.

All the statements made in this section and Section~\ref{sect4} are to be understood at the CE $\mathbf{x}$. And $\mathbf{x}$ will often be omitted if it does not cause confusion.\\

The following facts follow immediately from \eqref{cost_nonatom} and \eqref{cost_atom}. They will be repeatedly referred to in this work without further explanation.

\noindent{\sc Facts.}
\begin{align*}
 c_{r}(x_{r})=c^{0},\;\text{ if } r\, \in \mathcal{R}^{0};\quad &c_{r}(x_{r})\geq c^{0},\;\text{ if } r \in \mathcal{R}\setminus \mathcal{R}^{0}.\\
\forall\, k\in \mathcal{K},\quad \hat{c}^{k}_{r}(\mathbf{x})=\hat{c}^{k},\;\text{ if } r\, \in \mathcal{R}^{k};\quad & \hat{c}^{k}_{r}(\mathbf{x})\geq \hat{c}^{k},\;\text{ if } r \in \mathcal{R}\setminus \mathcal{R}^{k}.
\end{align*}

The following lemma states that an arc used by a coalition costs less than any arc not used by it.
\begin{lemma}\label{equation2}
For any coalition $k$, for any arc $r$ in $\mathcal{R}^{k}$ and any arc $s$ in $\mathcal{R}\setminus \mathcal{R}^{k}$, $c_{r}(x_{r})<c_{s}(x_{s})$.
\end{lemma}
\begin{proof} Since $x^{k}_{r}>0$ and $x^{k}_{s}=0$, $c_{r}(x_{r}) < c_{r}(x_{r})+x^{k}_{r}c'_{r}(x_{r}) = \hat{c}^{k} \leq c_{s}(x_{s})$.
\end{proof}

The next lemma shows that an arc used by individuals is also used by all the coalitions. Besides, the average cost to any coalition is not lower than the individuals' cost.
\begin{lemma}\label{me_flow_1}
For any coalition $k$,
\begin{enumerate}
 \item $\mathcal{R}^{0} \subset \mathcal{R}^{k}$, i.e. for all $r\in \mathcal{R}$, if $x^{0}_{r}>0$, then $x^{k}_{r}>0$.
 \item $c^{0}< \hat{c}^{k}$.
 \item $Y^{0} = \underline{Y}^{k} \leq Y^{k}$.
\end{enumerate}
\end{lemma}
\begin{proof} (i) Suppose that $x^{0}_{r}>0$. If $x^{k}_{r}=0$, there is another arc $s$ such that $x^{k}_{s}>0$. Then, $c_{r}(x_{r})\geq \hat{c}^{k}(\mathbf{x})= c_{s}(x_{s})+x^{k}_{s}c'_{s}(x_{s})>c_{s}(x_{s})$. However, $x^{0}_{r}>0$, hence $c_{r}(x_{r}) \leq c_{s}(x_{s})$, a contradiction.

(ii) Take $r$ in $\mathcal{R}^{0}$. By (i), $x^{k}_{r}>0$. Thus, $\hat{c}^{k}=c_{r}(x_{r})+x^{k}_{r}c'(x_{r})>c_{r}(x_{r})=c^{0}$.

(iii) The individuals take the arcs with the lowest cost, hence $Y^{0} \leq \underline{Y}^{k} \leq Y^{k}$. And (i) implies that $Y^{0}= \underline{Y}^{k}$.
\end{proof}

The next lemma states that an arc used by a coalition is also used by any larger coalition, and the larger one sends more flow on it.

\begin{lemma}\label{me_flow_2}
Let two coalitions $k$ and $l$ be such that $T^{k} < T^{l}$. Then the following are true.
  \begin{enumerate}
    \item $\mathcal{R}^{k}\subset \mathcal{R}^{l}$, i.e. for all $r\in \mathcal{R}$, if $x^{k}_{r}>0$, then $x^{l}_{r}>0$.
    \item $\hat{c}^{k} < \hat{c}^{l}$.
    \item For any arc $r$, $x^{k}_{r} \leq x^{l}_{r}$, and the inequality is strict if $x^{k}_{r}>0$.
    \item $Y^{k} \leq Y^{l}$, and the equality holds if, and only if, $Y^{l}=Y^{k}=Y^{0}$.
  \end{enumerate}

If $T^{k} = T^{l}$, all the inequalities or inclusions above become equalities.
\end{lemma}

\begin{proof}
(i) Suppose that $T^{k} <T^{l}$. If $\mathcal{R}^{k}\not \subset \mathcal{R}^{l}$, there is some $r$ such that $x^{k}_{r}>0$ but $x^{l}_{r}=0$. Hence, $\hat{c}^{k}=c_{r}(x_{r})+x^{k}_{r}c'_{r}(x_{r})>c_{r}(x_{r}) \geq \hat{c}^{l}$. In particular, $\hat{c}^{k}> \hat{c}^{l}$.

For all $s$ in $\mathcal{R}\setminus\mathcal{R}^{k}$, $c_{s}(x_{s}) \geq \hat{c}^{k} > \hat{c}^{l}$, which implies that $x^{l}_{s}=0$. In consequence, $\mathcal{R}\setminus\mathcal{R}^{k}\subset \mathcal{R}\setminus\mathcal{R}^{l}$ or, equivalently, $\mathcal{R}^{l}\subset \mathcal{R}^{k}$.

For all $r$ in $\mathcal{R}^{l}$ and, {\em a fortiori}, in $\mathcal{R}^{k}$, $\hat{c}^{k}=c_{r}(x_{r})+x^{k}_{r}c'_{r}(x_{r})$ and $\hat{c}^{l}=c_{r}(x_{r})+x^{l}_{r}c'_{r}(x_{r})$. Hence, $x^{k}_{r}-x^{l}_{r}=(\hat{c}^{k} - \hat{c}^{l})/c'_{r}(x_{r})>0$, so that $x^{k}_{r}>x^{l}_{r}$. As a result, $T^{l}=\sum_{r \in \mathcal{R}^{l}}x^{l}_{r}<\sum_{r \in \mathcal{R}^{l}}x^{k}_{r}\leq T^{k}$, a contradiction.

Therefore, $\mathcal{R}^{k}\subset \mathcal{R}^{l}$.

Suppose that $T^{k}=T^{l}$. The above proof is still valid. Thus, $\mathcal{R}^{k}\subset \mathcal{R}^{l}$ and, by symmetry, $\mathcal{R}^{l}\subset \mathcal{R}^{k}$. This leads to $\mathcal{R}^{k}= \mathcal{R}^{l}$.

(ii) and (iii) Suppose that $T^{k} <T^{l}$. By (i), $\mathcal{R}^{k}\subset \mathcal{R}^{l}$. There are two cases.

{\sc Case 1.} $\mathcal{R}^{k}=\mathcal{R}^{l}$. Given $r$ in $\mathcal{R}^{k}=\mathcal{R}^{l}$, $\hat{c}^{k}=c_{r}(x_{r})+x^{k}_{r}c'_{r}(x_{r})$ and $\hat{c}^{l}=c_{r}(x_{r})+x^{l}_{r}c'_{r}(x_{r})$. It follows that $x^{l}_{r}-x^{k}_{r}=(\hat{c}^{l} - \hat{c}^{k})/c'_{r}(x_{r})$. Thus, $0<T^{l} -T^{k}=\sum_{r \in \mathcal{R}^{k}}(x^{l}_{r}-x^{k}_{r})=(\hat{c}^{l} - \hat{c}^{k})\,\sum_{r \in \mathcal{R}^{k}}1/c'_{r}(x_{r})$ and, consequently, $\hat{c}^{l} > \hat{c}^{k}$.

{\sc Case 2.} $\mathcal{R}^{k}\subset \mathcal{R}^{l}$ but $\mathcal{R}^{k}\neq \mathcal{R}^{l}$. Take $s$ in $\mathcal{R}^{l}\setminus \mathcal{R}^{k}$. Then, $\hat{c}^{l}=c_{s}(x_{s})+x^{l}_{s}c'_{s}(x_{s})>c_{s}(x_{s}) \geq \hat{c}^{k}$.

In both cases, $\hat{c}^{l} > \hat{c}^{k}$. For all $r$ in $\mathcal{R}^{k}$, $x^{l}_{r}-x^{k}_{r}=(\hat{c}^{l} - \hat{c}^{k})/c'_{r}(x_{r})>0$; in particular, $x^{k}_{r}< x^{l}_{r}$. And for all $r$ in $\mathcal{R}\setminus \mathcal{R}^{k}$, $0=x^{k}_{r}\leq x^{l}_{r}$.

Suppose that $T^{k} =T^{l}$.  By (i), $\mathcal{R}^{k}= \mathcal{R}^{l}$. On the one hand, the same argument as for Case 1 leads to $\hat{c}^{l} =\hat{c}^{k}$ and $x^{l}_{r}=x^{k}_{r}$ for all $r$ in $\mathcal{R}^{k}= \mathcal{R}^{l}$. On the other hand, $x^{l}_{r}=x^{k}_{r}=0$ for all $r$ in $\mathcal{R}\setminus \mathcal{R}^{k}$.

(iv) Suppose that $T^{k} < T^{l}$. According to (i), $\mathcal{R}^{k}\subset \mathcal{R}^{l}$.

Set $\tilde{Y}^{l}=\sum_{r \in \mathcal{R}^{k}}x^{l}_{r}c_{r}(x_{r})/\sum_{r \in \mathcal{R}^{k}}x^{l}_{r}$, the average cost to coalition $l$ on $\mathcal{R}^{k}$. By Lemma~\ref{equation2}, the arcs in $\mathcal{R}^{k}$ cost strictly less than those in $\mathcal{R}\setminus\mathcal{R}^{k}$. One deduces that $\tilde{Y}^{l}=Y^{l}$ if $\mathcal{R}^{k}$ is equal to $\mathcal{R}^{l}$, and $\tilde{Y}^{l}<Y^{l}$ if $\mathcal{R}^{k}$ is a proper subset of $\mathcal{R}^{l}$.

Now, let us show that $Y^{k}\leq \tilde{Y}^{l}$.
\begin{align*}
\tilde{Y}^{l}
 & =\frac{\sum_{r \in \mathcal{R}^{k}}x^{l}_{r}c_{r}(x_{r})}{\sum_{r \in \mathcal{R}^{k}}x^{l}_{r}} \\
 & =\frac{\sum_{r \in \mathcal{R}^{k}}x^{k}_{r}c_{r}(x_{r})+\sum_{r \in \mathcal{R}^{k}}(x^{l}_{r}-x^{k}_{r})c_{r}(x_{r})}{\sum_{r \in \mathcal{R}^{k}}x^{k}_{r}+\sum_{r \in \mathcal{R}^{k}}(x^{l}_{r}-x^{k}_{r})} \\
 & =\frac{Y^{k}T^{k}+\sum_{r \in \mathcal{R}^{k}}(x^{l}_{r}-x^{k}_{r})c_{r}(x_{r})}{T^{k}+\sum_{r \in \mathcal{R}^{k}}(x^{l}_{r}-x^{k}_{r})}.
\end{align*}

It follows from (iii) that, for all $r$ in $\mathcal{R}^{k}$, $x^{l}_{r}-x^{k}_{r}>0$. The relation $Y^{k}\leq \tilde{Y}^{l}$ is thus equivalent to the inequality
\begin{equation}\label{equation1}
Y^{k}\leq \frac{\sum_{r \in \mathcal{R}^{k}}(x^{l}_{r}-x^{k}_{r})c_{r}(x_{r})}{\sum_{r \in \mathcal{R}^{k}}(x^{l}_{r}-x^{k}_{r})}.
\end{equation}

For $r$ in $\mathcal{R}^{k}$, $x^{l}_{r}-x^{k}_{r}=(\hat{c}^{l} - \hat{c}^{k})/c'_{r}(x_{r})$ and $c_{r}(x_{r})=\hat{c}^{k}-x^{k}_{r}c'_{r}(x_{r})$. Inequality~\eqref{equation1} can thus be written as
\begin{align}
 &\frac{\sum_{r \in \mathcal{R}^{k}}x^{k}_{r}c_{r}(x_{r})}{\sum_{r \in \mathcal{R}^{k}}x^{k}_{r}}
\leq \frac{\sum_{r \in \mathcal{R}^{k}}c_{r}(x_{r})(\hat{c}^{l} - \hat{c}^{k})/c'_{r}(x_{r})}{\sum_{r \in \mathcal{R}^{k}}(\hat{c}^{l} - \hat{c}^{k})/c'_{r}(x_{r})}
=\frac{\sum_{r \in \mathcal{R}^{k}}c_{r}(x_{r})/c'_{r}(x_{r})}{\sum_{r \in \mathcal{R}^{k}}1/c'_{r}(x_{r})}\nonumber\\
\Leftrightarrow \;&
\sum_{r \in \mathcal{R}^{k}}x^{k}_{r}c_{r}(x_{r})\sum_{r \in \mathcal{R}^{k}}\frac{1}{c'_{r}(x_{r})}\leq \sum_{r \in \mathcal{R}^{k}}x^{k}_{r}\sum_{r \in \mathcal{R}^{k}}\frac{c_{r}(x_{r})}{c'_{r}(x_{r})}\nonumber\nonumber\\
\Leftrightarrow \;&
\sum_{r \in \mathcal{R}^{k}}x^{k}_{r}\left(\hat{c}^{k}-x^{k}_{r}c'_{r}(x_{r})\right)\sum_{r \in \mathcal{R}^{k}}\frac{1}{c'_{r}(x_{r})}\leq \sum_{r \in \mathcal{R}^{k}}x^{k}_{r}\sum_{r \in \mathcal{R}^{k}}\frac{\hat{c}^{k}-x^{k}_{r}c'_{r}(x_{r})}{c'_{r}(x_{r})}\nonumber \\
\Leftrightarrow \;&
\sum_{r \in \mathcal{R}^{k}}\left(x^{k}_{r}\right)^{2}\, c'_{r}(x_{r})\sum_{r \in \mathcal{R}^{k}}\frac{1}{c'_{r}(x_{r})}\geq \Bigl(\sum_{r \in \mathcal{R}^{k}}x^{k}_{r}\Bigr)^{2}.\label{cauchy-schwarz}
\end{align}

Inequality~\eqref{cauchy-schwarz} follows from Cauchy-Schwarz inequality. Furthermore, the equality holds (or, equivalently, $Y^{k}= \tilde{Y}^{l}$) if, and only if, $x^{k}_{r}c'_{r}(x_{r})$ is constant for all $r$ in $\mathcal{R}^{k}$. When this is the case,  $c_{r}(x_{r})=\hat{c}^{k}-x^{k}_{r}c'_{r}(x_{r})$ is also a constant for all $r$ in $\mathcal{R}^{k}$. According to Lemma~\ref{me_flow_1} (iii), this constant must be equal to $c^{0}$.

The relations $Y^{k}(\mathbf{x})\leq \tilde{Y}^{l}(\mathbf{x})\leq Y^{l}(\mathbf{x})$ is now established. Suppose, moreover, that $Y^{k}(\mathbf{x})=Y^{l}(\mathbf{x})$. On the one hand, $\tilde{Y}^{l}(\mathbf{x})=Y^{l}(\mathbf{x})$, implying that $\mathcal{R}^{k}= \mathcal{R}^{l}$. On the other hand, $Y^{k}(\mathbf{x})=\tilde{Y}^{l}(\mathbf{x})$, implying that every arc in $\mathcal{R}^{k}$ costs $c^{0}$.

Suppose that $T^{k} = T^{l}$. The result follows directly from (iii).
\end{proof}
\begin{remark}\rm
 (i) and (iii) of Lemma~\ref{me_flow_2} were also proven by Orda, Rom and Shimkin \cite{Orda93} with another formulation for atomic splittable games. Lemma~1 in \cite{Orda93} claims that, at the NE, if $x^{k}_{r}<x^{l}_{r}$ for some arc $r$, then $x^{k}_{s}\leq x^{l}_{s}$ for all arc $s$, and the inequality is strict if $x^{k}_{s}>0$.
\end{remark}

The following corollary of Lemma~\ref{me_flow_2} shows that the behavior of a coalition at the CE is specified by its weight.
\begin{corollary}\label{me_flow_3}
Two coalitions send the same weight on every arc if, and only if, they have the same weight. In this case, they have the same average cost.
\end{corollary}

\section{Comparison between CE and WE}\label{sect4}
The previous section was contributed to the basic properties of the CE $\mathbf{x}$ of the game $\Gamma(\mathcal{R},\mathbf{c},\mathbf{T})$. This section will compare it with the WE $\mathbf{w}=(w_{r})_{r\in \mathcal{R}}$ of the corresponding nonatomic game $\Gamma(\mathcal{R},\mathbf{c})$. The equilibrium cost at $\mathbf{w}$ is denoted by $W \in \mathbb{R}$. One says that $\mathbf{x}$ {\em induces} $\mathbf{w}$ if $\mathbf{x}'=\mathbf{w}$, i.e. $x_{r}=w_{r}$ for all $r\in \mathcal{R}$.

Following Hayrapetyan, Tardos and Wexler~\cite{Hay06}, let $\mathcal{R}_{-}=\{r\in \mathcal{R}\,|\, x_{r}<w_{r}\}$, $\mathcal{R}_{+}=\{r\in \mathcal{R}\,|\,  x_{r}>w_{r}\}$ and $\mathcal{R}_{J}=\{r\in \mathcal{R}\,|\,  x_{r}=w_{r}\}$ be, respectively, the set of {\em underloaded arcs}, the set of {\em overloaded arcs} and the set of {\em justly-loaded arcs}.

\begin{lemma}\label{over_underloaded}
If $\mathbf{x}$ does not induce $\mathbf{w}$, then the following are true.
\begin{enumerate}
  \item For all $s \in \mathcal{R}_{-}$ and for all $r \in \mathcal{R}_{+}$, $c_{s}(x_{s})<W<c_{r}(x_{r})$.
  \item $\mathcal{R}^{0}\subset \mathcal{R}_{-}$, i.e. for all $r\in \mathcal{R}$, if $x^{0}_{r}>0$, then $x_{r} < w_{r}$.
  \item  $\mathcal{R}_{+}\subset \mathcal{R}^{1}$, i.e. for all $r\in \mathcal{R}$, if $x_{r}>w_{r}$, then $x^{1}_{r} >0$.
\end{enumerate}
\end{lemma}
\begin{proof}
(i) As $\mathbf{x}'\neq \mathbf{w}$, both $\mathcal{R}_{-}$ and $\mathcal{R}_{+}$ are nonempty. Take $s$ in $\mathcal{R}_{-}$ and $r$ in $\mathcal{R}_{+}$, then $w_{s} > x_{s}\geq 0$ and $w_{r} < x_{r}$. In particular, $w_{s}> 0$, which implies that $s$ is used at the WE. Then, $c_{s}(x_{s})<c_{s}(w_{s})= W\leq c_{r}(w_{r})<c_{r}(x_{r})$.

(ii) The individuals take the arcs of the lowest cost at $\mathbf{x}$. According to (i), these arcs must be in $\mathcal{R}_{-}$, hence $\mathcal{R}^{0}\subset \mathcal{R}_{-}$.

(iii) For all $r$ in $\mathcal{R}_{+}$, since $x_{r}>w_{r}\geq 0$, $r$ is used at $\mathbf{x}$. According to Lemma~\ref{me_flow_1} and Lemma~\ref{me_flow_2}, it is used by the largest coalition, coalition 1. Thus, $\mathcal{R}_{+}\subset \mathcal{R}^{1}$.
\end{proof}

The following theorem compares the equilibrium costs at $\mathbf{x}$ with the equilibrium cost at $\mathbf{w}$.

\begin{theorem}\label{main}
If $\mathbf{x}$ does not induce $\mathbf{w}$, then $Y^{0}(\mathbf{x}) < W$ and $Y^{k}(\mathbf{x}) < W$ for each coalition $k$. Consequently, $Y(\mathbf{x})  < W$.
\end{theorem}
\begin{proof}
For the individuals, for all $r \in \mathcal{R}^{0}$, $c_{r}(x_{r})=Y^{0}(\mathbf{x})$. Lemma~\ref{over_underloaded}(ii) implies that $r$ is in $\mathcal{R}_{-}$, and Lemma~\ref{over_underloaded}(i) shows that $c_{r}(x_{r})<W$.

For the coalitions, it is enough to show that $Y^{1}(\mathbf{x}) < W$ for the largest coalition, coalition 1. Once this is proven, the remaining results follow from Lemma~\ref{me_flow_2}.

Let us define an auxiliary flow $\mathbf{z}$ in $F$, such that it induces $\mathbf{w}$ and satisfies the following conditions.
\begin{equation*}
\begin{cases}
\,z^{1}_{r} > x^{1}_{r},\;\; z^{k}_{r} \geq x^{k}_{r},\;\; z^{0}_{r} = x^{0}_{r},  \quad & \;k\in \mathcal{K}\setminus \{1\},\;\; r\in \mathcal{R}_{-}, \\
\,z^{1}_{r} < x^{1}_{r},\; \;z^{k}_{r} \leq x^{k}_{r},\;\; z^{0}_{r} = x^{0}_{r},  \quad & \;k\in \mathcal{K}\setminus \{1\},\;\; r\in \mathcal{R}_{+}, \\
\,z^{k}_{r} = z^{k}_{r}, \quad & \;k=0 \text{ or } k \in \mathcal{K},\;\; r\in \mathcal{R}_{J}.
\end{cases}
\end{equation*}
For example, one can define, for all $k\in \mathcal{K}$ and $r \in \mathcal{R}_{+}$, $z^{k}_{r}=x^{k}_{r}-d^{k}_{r}$, where $d^{k}_{r}=(x_{r}-w_{r})x^{k}_{r}/\sum_{l\in \mathcal{K}}x^{l}_{r}$, while for all $r \in \mathcal{R}_{-}$, $z^{k}_{r}=x^{k}_{r}+d^{k}_{r}$, where $d^{k}_{r}=(w_{r}-x_{r})\sum_{t\in\mathcal{R}_{+}}d^{k}_{t}/\sum_{s\in\mathcal{R}_{-}}(w_{s}-x_{s})$. The above conditions are satisfied due to Lemma~\ref{over_underloaded}(iii).

Let us define another auxiliary flow $\mathbf{y}$ in $F$ as follows. For all $r\in \mathcal{R}$,
\begin{equation*}
y^{k}_{r} =
\begin{cases} \,z^{k}_{r}, \; &\text{ if } k \neq 1,\\
\,x^{1}_{r}, \; &\text{ if } k = 1.
\end{cases}
\end{equation*}
In other words, at $\mathbf{y}$, the individuals and all coalitions, except coalition $1$, behave like at $\mathbf{w}$, while coalition $1$ behaves like at $\mathbf{x}$. Let us show that $u^{1}(\mathbf{x})\leq u^{1}(\mathbf{y}) < u^{1}(\mathbf{z})$.

Some preliminary results are needed.

For all $s \in \mathcal{R}_{-}$ and for all $r \in \mathcal{R}_{+}$,

i) $y_{s}\geq x_{s}$, because $y_{s}- x_{s}=\sum_{\mathcal{K}\setminus \{1\}}z^{k}_{s}-x^{k}_{s}\geq 0$.

ii) $x_{r}\geq y_{r}$, because $x_{r}- y_{r}=\sum_{\mathcal{K}\setminus \{1\}}x^{k}_{r}-z^{k}_{r}\geq 0$.

iii) By Lemma~\ref{over_underloaded}(iii), coalition $1$ takes arc $r$. Thus,
\begin{equation}\label{cost_atom_bis}
c_{r}(x_{r})+x^{1}_{r}c'_{r}(x_{r}) =\hat{c}^{1}(\mathbf{x}) \leq c_{s}(x_{s})+x^{1}_{s}c'_{s}(x_{s}).
\end{equation}

Moreover, due to Assumption~\ref{cost_assumption_1}, for all $0 < x <x_{r}$ and $y > 0$,
\begin{equation}\label{compare_2}
c_{r}(x_{r}-x)+(x^{1}_{r}-x)\,c'_{r}(x_{r}-x) < \hat{c}^{1}(\mathbf{x}) < c_{s}(x_{s}+y)+(x^{1}_{s}+y)\,c'_{s}(x_{s}+y).
\end{equation}

Finally, $c_{s}(x_{s})< c_{r}(x_{r})$ by Lemma~\ref{over_underloaded}(i). Then, it follows from \eqref{cost_atom_bis} that $x^{1}_{r}c'_{r}(x_{r}) < x^{1}_{s}c'_{s}(x_{s})$. Let $B$ be a constant such that $\max_{t \in \mathcal{R}_{+}}\{x^{1}_{t}\,c'_{t}(x_{t})\} \leq B \leq \min_{t \in \mathcal{R}_{-}}\{x^{1}_{t}\,c'_{t}(x_{t})\}$. Then, by Assumption~\ref{cost_assumption_1}, for all $x$ and $y$ such that $0\leq x < x_{r}$ and $y > x_{s}$,
\begin{equation}\label{compare_1}
x^{1}_{r}\,c'_{r}(x) \leq B \leq x^{1}_{s}\,c'_{s}(y).
\end{equation}

Now, let us show that $u^{1}(\mathbf{x})\leq u^{1}(\mathbf{y}) < u^{1}(\mathbf{z})$.
\begin{align*}
u^{1}(\mathbf{y}) - u^{1}(\mathbf{x})\,
& =\, \sum_{r\in \mathcal{R}}y^{1}_{r}c_{r}(y_{r})-\sum_{r\in \mathcal{R}}x^{1}_{r}c_{r}(x_{r}) = \sum_{r\in \mathcal{R}}x^{1}_{r}c_{r}(y_{r})-\sum_{r\in \mathcal{R}}x^{1}_{r}c_{r}(x_{r})\\
& =\, \sum_{s\in \mathcal{R}_{-}}\bigl[\,x^{1}_{s}c_{s}(y_{s})-x^{1}_{s}c_{s}(x_{s})\,\bigr]-\sum_{r\in \mathcal{R}_{+}}\bigl[\,x^{1}_{r}c_{r}(x_{r})-x^{1}_{r}c_{r}(y_{r})\,\bigr]\\
& =\, \sum_{s\in \mathcal{R}_{-}}\int_{x_{s}}^{y_{s}}x^{1}_{s}\,c'_{s}(x)\,\text{d}x - \sum_{r\in \mathcal{R}_{+}}\int_{y_{r}}^{x_{r}}x^{1}_{r}\,c'_{r}(x)\,\text{d}x\nonumber\\
& \geq \, \sum_{s\in \mathcal{R}_{-}}(y_{s}-x_{s})B-\sum_{r\in \mathcal{R}_{+}}(x_{r}-y_{r})B = \sum_{r\in \mathcal{R}}(y_{r}-x_{r})B  = 0.
\end{align*}
The inequality above is due to \eqref{compare_1} and the fact that $y_{s}\geq x_{s}$ for all $s$ in $\mathcal{R}_{-}$ and $x_{r}\geq y_{r}$ for all $r$ in $\mathcal{R}_{+}$.
\begin{align}
u^{1}(\mathbf{z}) - u^{1}(\mathbf{y})\nonumber\,&=\, \sum_{r\in \mathcal{R}}z^{1}_{r}c_{r}(w_{r})-\sum_{r\in \mathcal{R}}y^{1}_{r}c_{r}(y_{r})\,=\, \sum_{r\in \mathcal{R}}z^{1}_{r}c_{r}(w_{r})-\sum_{r\in \mathcal{R}}x^{1}_{r}c_{r}(w_{r}-z^{1}_{r}+x^{1}_{r})\nonumber\\
&=\, \sum_{s\in \mathcal{R}_{-}}\bigl[\,z^{1}_{s}c_{s}(w_{s})-x^{1}_{s}c_{s}(w_{s}-z^{1}_{s}+x^{1}_{s})\,\bigr]-\sum_{r\in \mathcal{R}_{+}}\bigl[\,x^{1}_{r}c_{r}(w_{r}-z^{1}_{r}+x^{1}_{r})-z^{1}_{r}c_{r}(w_{r})\,\bigr]\nonumber\\
\displaybreak[3]
&=\, \sum_{s\in \mathcal{R}_{-}}\int_{x^{1}_{s}}^{z^{1}_{s}} \frac{\partial}{\partial x}\bigl[\,x {c}_{s}(w_{s}-z^{1}_{s}+x)\,\bigr]\,\text{d}x- \sum_{r\in \mathcal{R}_{+}}\int_{z^{1}_{r}}^{x^{1}_{r}}\frac{\partial}{\partial x}\bigl[\,x {c}_{r}(w_{r}-z^{1}_{r}+x)\,\bigr]\,\text{d}x\nonumber\\
&=\, \sum_{s\in \mathcal{R}_{-}}\int_{x^{1}_{s}}^{z^{1}_{s}} \bigl[\,{c}_{s}(w_{s}-z^{1}_{s}+x)+x c'_{s}(w_{s}-z^{1}_{s}+x)\,\bigr]\,\text{d}x\nonumber\\
&\quad - \sum_{r\in \mathcal{R}_{+}}\int_{z^{1}_{r}}^{x^{1}_{r}}\bigl[\,{c}_{r}(w_{r}-z^{1}_{r}+x)+x c'_{r}(w_{r}-z^{1}_{r}+x)\,\bigr]\,\text{d}x\nonumber\\
&\geq \, \sum_{s\in \mathcal{R}_{-}}\int_{x^{1}_{s}}^{z^{1}_{s}} \bigl[\,{c}_{s}(x_{s}-x^{1}_{s}+x)+x c'_{s}(x_{s}-x^{1}_{s}+x)\,\bigr]\,\text{d}x \label{ineq_4}\\
&\quad - \sum_{r\in \mathcal{R}_{+}}\int_{z^{1}_{r}}^{x^{1}_{r}}\bigl[\,{c}_{r}(x_{r}-x^{1}_{r}+x)+x c'_{r}(x_{r}-x^{1}_{r}+x)\,\bigr]\,\text{d}x\nonumber\\
&>\, \sum_{s\in \mathcal{R}_{-}}(z^{1}_{s}-x^{1}_{s})\,\hat{c}^{1}(\mathbf{x})-\sum_{r\in \mathcal{R}_{+}}(x^{1}_{r}-z^{1}_{r})\,\hat{c}^{1}(\mathbf{x})\label{ineq_6}\\
&=\,\sum_{r\in \mathcal{R}}(z^{1}_{r}-x^{1}_{r})\,\hat{c}^{1}(\mathbf{x})=(T^{1}-T^{1})\,\hat{c}^{1}(\mathbf{x})= 0.\nonumber
\end{align}
Inequality~\eqref{ineq_4} is due to the following facts which follow immediately from the definition of $\mathbf{z}$. For $s$ in $\mathcal{R}_{-}$, $z^{1}_{s}>x^{1}_{s}$ and $w_{s}-z^{1}_{s}\geq x_{s}-x^{1}_{s}$, while for $r$ in $\mathcal{R}_{+}$, $x^{1}_{r}>z^{1}_{r}$ and $x_{r}-x^{1}_{r}\geq w_{r}-z^{1}_{r}$. Inequality~\eqref{ineq_6} is due to \eqref{compare_2}.

Thus, one has proved that $u^{1}(\mathbf{x})< u^{1}(\mathbf{z})$ or, equivalently, $Y^{1}(\mathbf{x}) < Y^{1}(\mathbf{z})$. Besides, since $\mathbf{z}$ induces $\mathbf{w}$, every arc used at $\mathbf{z}$ costs $W$. In consequence, $Y^{1}(\mathbf{z})=W$, which completes the proof.
\end{proof}
\begin{remark}\rm
Cominetti, Correa and Stier-Moses~\cite[Corollary 4.1]{Com09} proved that, if the cost functions are non-decreasing, convex and differentiable, and the atomic splittable players are identical, then the social cost at any NE in an atomic splittable game is lower than that at the corresponding WE. Hayrapetyan, Tardos and Wexler~\cite[Theorem 2.3]{Hay06} proved that, in a two-terminal parallel-arc network, if the cost functions are non-decreasing, convex and differentiable, then the social cost at any NE in an atomic splittable game is lower than that at the corresponding WE. In this work, stronger convexity conditions on the cost functions allow to prove that not only the social average cost, but also the average cost to any coalition and the individuals' cost are lower at the CE than at the WE.
\end{remark}
\begin{remark}\rm
Cominetti, Correa and Stier-Moses~\cite[\S 2.1]{Com09} provided an example where two groups of individuals have {\em different} origin/destination pairs. They showed that, when one of the two groups forms a coalition, both the social cost and the average cost to this coalition are increased. This implies that further studies are needed for more general cases where the network is not two-terminal parallel-arc type.
\end{remark}

\section{Impact of the composition of the population on the CE costs}\label{sect5}
This section focuses on the relation between the costs at the CE and the composition of the set of the players, i.e. its partition into coalitions and individuals. In the first part, one considers a unique coalition of weight $T\in [\,0,1]$, and studies the variation of the coalition's cost and the remaining individuals' cost with respect to $T$. In the second part, for a general composition of the set of the players, one shows that, whenever a coalition decreases, i.e. some of its members become individuals, the individuals' cost is increased.

\subsection{CE costs as functions of the size of the unique coalition}
Suppose that a unique coalition of weight $T\in [0,1]$ is formed.
\begin{figure}[!htbp]
\begin{center}
\begin{tikzpicture}
\draw[thin,->] (0,0)node[below left,scale=0.7]{0} -- (4.7,0)node[below,scale=0.7] {$T^{1}$};
\draw[thin,->] (0,0) -- (0,4.7)node[left,scale=0.7] {$T^{0}$};

\fill[gray,opacity=0.1] (4,0) -- (4,4) -- (0,4) -- cycle;
\draw (4,4) -- (4,0) node[below,scale=0.7] {$1$};
\draw [dotted] (0,4) -- (4,0);
\draw [dashed, thick] (4,4) -- (0,4) node[left,scale=0.7] {$1$} ;

\draw [thick] (2.5,1.5) -- (0,1.5) node[left,scale=0.7] {$1-T$};
\draw [thick] (0,0) -- (4,0) ;
\draw [dotted] (2.5,1.5) -- (2.5,0) node[below,scale=0.7] {$T$};
\draw [dashed, thick] (2.5,1.5) -- (4,1.5);

\node[scale=0.7] (WE)at(1.8,4.5) {$WE$};
\node[scale=0.7] (SO)at(1.8,-0.5) {$SO$};
\node[scale=0.7] (CE)at(1.5,1.1) {$CE(1-T;T)$};
\draw [->,thin] (WE) to (2,4.05);
\draw [->,thin] (SO) to (2,-0.05);
\draw [->,thin] (CE) to (2,1.45);
\end{tikzpicture}
\caption{\label{coalition_size} Composition of the players}
\end{center}
\end{figure}
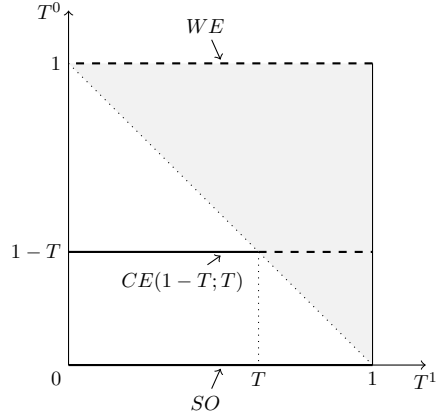
Every horizontal line in Figure~\ref{coalition_size} represents a composition of the set of the players: the unique coalition is presented by the plain part on the left, and the individuals by the dashed part on the right. From bottom to top, the unique coalition decreases. The top (dashed) line stands for the WE $\mathbf{w}$, the bottom (plain) line stands for the social optimum, and any horizontal line between them stands for the CE of a one-coalition composite game.
\begin{lemma}\label{unique_coalition_T}
There exists a number $\tilde{T}$ in $[\,0,\,1]$ such that the CE in $\Gamma(\mathcal{R},\mathbf{c},(1-T;\,T))$ induces $\mathbf{w}$ if, and only if, $T \leq \tilde{T}$.
\end{lemma}
\begin{proof}
Let $\mathcal{R}_{a}=\{r\in \mathcal{R}\,|\, w_{r}>0\}$ be the set of used arcs at $\mathbf{w}$, and $\mathcal{R}_{i}=\mathcal{R}\backslash \mathcal{R}_{a}=\{r\in \mathcal{R}\,|\, w_{r}=0\}$ the set of unused arcs, which may be empty. Set $A=\sum_{r\in \mathcal{R}_{a}}\frac{1}{c'_{r}(w_{r})}$. Then, the following constant
\begin{equation}\label{def:T}
\tilde{T}=\min \Bigl\lbrace \min_{r \in \mathcal{R}_{a}} w_{r}c'_{r}(w_{r})A,\;\min_{r \in \mathcal{R}_{i}}\left(c_{r}(0)-W\right)A\Bigr\rbrace,
\end{equation}
is the threshold. This can be proven in two cases.

{\sc Case 1.} For all $r \in \mathcal{R}_{i}$, $c_{r}(0)>W$.

In this case, $\tilde{T}>0$. Let us show that the CE induces $\mathbf{w}$ if, and only if, $T\leq \tilde{T}$.

On the one hand, if $T\leq \tilde{T}$, the following flow $\mathbf{x}$ is the CE.
\begin{equation*}
\begin{cases}
\, x^{1}_{r} = \frac{T}{Ac'_{r}(w_{r})},\quad x^{0}_{r}= w_{r}-x^{1}_{r},\hspace{1cm} & r \in \mathcal{R}_{a};\\
\, x^{1}_{r} = x^{0}_{r} = 0, & r \in \mathcal{R}_{i}.
\end{cases}
\end{equation*}
Indeed, $\mathbf{x}$ is well-defined because of the definition of $\tilde{T}$, and $\sum_{r\in \mathcal{R}}x^{1}_{r}=T$. Next, as $x_{r}=w_{r}$ for all $r$, the individuals do take the arcs of the lowest cost. Finally, it follows from the definition of $\tilde{T}$ that, for all $r \in \mathcal{R}_{i}$, $(c_{r}(0)-W)A\geq \tilde{T} \geq T$. It is not difficult to see that, for all $r\in \mathcal{R}_{a}$, $c_{r}(x_{r})+x^{1}_{r}c'_{r}(w_{r})=W+\frac{T}{A}$ and, for all $r\in \mathcal{R}_{i}$, $c_{r}(0)\geq W+\frac{T}{A}$. The equilibrium condition \eqref{mixed-cond-coal} is thus satisfied for the coalition. One deduces that $\mathbf{x}$ is the CE. Besides, it induces a WE, because $x_{r}=w_{r}$ for all $r$.

On the other hand, if the CE $\mathbf{x}$ induces $\mathbf{w}$, i.e. $x_{r}=w_{r}$ for all $r$, then for all $r \in \mathcal{R}_{i}$, $x_{r}=w_{r}=0$, which implies that there exists an arc $s\in \mathcal{R}_{a}$ such that $x^{1}_{s}>0$. However, for all $r\in \mathcal{R}_{a}$ such that $x^{1}_{r}=0$, one has $W=c_{r}(w_{r})=c_{r}(x_{r})\geq \hat{c}^{1}(\mathbf{x})= c_{s}(x_{s})+x^{1}_{s}c'_{s}(x_{s})>c_{s}(x_{s})=c_{s}(w_{s})=W$, a contradiction. Therefore, for all $r \in \mathcal{R}_{a}$, $x^{1}_{r}>0$ and $ c_{r}(x_{r})+x^{1}_{r}c'_{r}(x_{r})=\hat{c}^{1}(\mathbf{x})$. In consequence, $x^{1}_{r}=\frac{\hat{c}^{1}(\mathbf{x})-c_{r}(x_{r})}{c'_{r}(x_{r})}=\frac{\hat{c}^{1}(\mathbf{x})-W}{c'_{r}(w_{r})}$. The constraint $x^{1}_{r}\leq w_{r}$ implies that $\hat{c}^{1}-W<w_{r}c'_{r}(w_{r})$. Consequently, $T=\sum_{r\in \mathcal{R}_{a}}x^{1}_{r}=\sum_{r\in \mathcal{R}_{a}}\frac{\hat{c}^{1}-W}{c'_{r}(w_{r})}=(\hat{c}^{1}-W)A<w_{r}c'_{r}(w_{r})$, for all $r \in \mathcal{R}_{a}$.

Besides, for all $r \in \mathcal{R}_{i},\,\hat{c}^{1}\leq c_{r}(0)$, which implies that $T=(\hat{c}^{1}-W)A\leq (c_{r}(0)-W)A$. Thus, $T\leq \tilde{T}$ is proven.

{\sc Case 2.} There exists some $t \in \mathcal{R}_{i}$ such that $c_{t}(0)=W$.

In this case, $\tilde{T}=0$. Let us show that the CE does not induce a WE as long as $T>0$. Otherwise, suppose that for some $T>0$, the CE $\mathbf{x}$ induces a WE. By the same reasoning as in Case 1, there exists an arc $s\in \mathcal{R}_{a}$ such that $x^{1}_{s}>0$. Then, $c_{t}(0)=W=c_{s}(w_{s})< c_{s}(w_{s})+x^{1}_{s}c'_{s}(w_{s})=\hat{c}^{1}$. However, $x_{t}=w_{t}=0$, which implies that $c_{t}(0)\geq \hat{c}^{1}$, a contradiction.
\end{proof}
\begin{example} \label{ex:1}\rm
There are two parallel arcs $r_{1}$ and $r_{2}$, whose cost functions are, respectively, $c_{1}(x)=x+10$ and $c_{2}(x)=10x+1$. A computation shows that the threshold is $\tilde{T}=\frac{1}{10}$. If less than one tenth of the players join the coalition, the coalition changes actually nothing in the game equilibrium.
\end{example}
\begin{theorem}\label{efficiency_size}
Let $\tilde{T}$ be defined by \eqref{def:T}. The individuals' cost $Y^{0}(T)$, the average cost to the unique coalition $Y^{1}(T)$, and the social cost $Y(T)$ in $\Gamma(\mathcal{R},\mathbf{c},(1-T;\,T))$ have the following properties.
\begin{enumerate}
  \item For $T \in [\,0,\, \tilde{T}\,]$, $Y^{0}(T)=Y^{1}(T)=Y(T)=W$.
  \item For $T \in (\,\tilde{T},\, 1\,]$, $Y^{0}(T)<Y^{1}(T)<W$, $Y(T)<W$. In particular, $Y^{1}(1)<Y^{0}(0)=W$.
  \item $Y^{0}(T)$, $Y^{1}(T)$ and $Y(T)$ are all strictly decreasing with respect to $T$ on $[\,\tilde{T},\, 1\,]$.
\end{enumerate}
\end{theorem}
\begin{proof} First, notice that $Y^{1}(T)$ is not defined for $T=0$, and that $Y^{0}(T)$ is not defined for $T=1$. However, as the cost functions satisfy Assumptions~\ref{cost_assumption_1}, one can extend $Y^{1}(T)$ to $T=0$ and $Y^{0}(T)$ to $T=1$ without difficulty.

(i) See Lemma~\ref{unique_coalition_T}.

(ii) According to Lemma~\ref{unique_coalition_T}, when $\tilde{T}< T \leq 1$, the CE does not induce $\mathbf{w}$. Therefore, according to Lemma~\ref{me_flow_1}(iii) and Theorem~\ref{main}, $Y^{0}(T)\leq Y^{1}(T)<W$. It remains to show that $Y^{0}(T)< Y^{1}(T)$.

Suppose that $Y^{0}(T) = Y^{1}(T)$. Then, $Y^{0}(T)= \underline{Y}^{1}(T)= Y^{1}(T)$ by Lemma~\ref{me_flow_1}, where $\underline{Y}^{1}(T)$ is the lowest cost of the arcs used by the coalition. This means that every arc used by the coalition has the lowest cost $Y^{0}(T)$. Therefore, the CE does induce $\mathbf{w}$, a contradiction.

(iii) Suppose that $\tilde{T} < S < T \leq 1$. Let $\mathbf{x}=(\mathbf{x}^{0},\,\mathbf{x}^{1})$ and $\mathbf{y}=(\mathbf{y}^{0},\,\mathbf{y}^{1})$ be, respectively, the CE of the game $\Gamma(\mathcal{R},\mathbf{c},(1-T;\,T))$ and that of the game $\Gamma(\mathcal{R},\mathbf{c},(1-S;\,S))$. The other notations are as listed at the beginning of Section~\ref{sect3}.

Let $\mathcal{R}_{-}$ be the set of underloaded or justly-loaded arcs, and $\mathcal{R}_{+}=\mathcal{R}\setminus \mathcal{R}_{-}$ be the set of overloaded arcs. In other words,
\begin{equation*}
 \mathcal{R}_{-}=\{r\in \mathcal{R}\,|\,y_{r}\leq x_{r}\},\quad  \mathcal{R}_{+}=\{r\in \mathcal{R}\,|\,y_{r}>x_{r}\}.
\end{equation*}

If $\mathcal{R}_{+}=\emptyset$, then $\mathbf{x}'=\mathbf{y}'$, i.e. $x_{r}=y_{r}$ for all $r\in \mathcal{R}$. Let us prove that this is impossible.

Define $\mathcal{R}^{\flat}=\{r\in \mathcal{R}\,|\,c_{r}(x_{r})=c^{0}(\mathbf{x})\}$ and $\mathcal{R}^{\sharp}=\{r\in \mathcal{R}\,|\,x_{r}>0,\,c_{r}(x_{r})>c^{0}(\mathbf{x})\}$. Then, for all $r\in \mathcal{R}^{\sharp}, x^{0}_{r}=0$. As $T>S>\tilde{T}$, according to (ii), $\mathbf{x}$ and $\mathbf{y}$ do not induce $\mathbf{w}$. Therefore, $\mathcal{R}^{\sharp}$ is nonempty.

For all $r\in \mathcal{R}^{\flat}$, $c_{r}(y_{r})=c_{r}(x_{r})=c^{0}(\mathbf{x})$ while, for all $r\in \mathcal{R}^{\sharp}$, $c_{r}(y_{r})=c_{r}(x_{r})>c^{0}(\mathbf{x})$. Hence, $c^{0}(\mathbf{x})$ is the minimal arc cost at $\mathbf{y}$. One deduces that $c^{0}(\mathbf{y})=c^{0}(\mathbf{x})$ and, for all $r\in \mathcal{R}^{\sharp}$, $y^{0}_{r}=0, y^{1}_{r}=y_{r}=x_{r}$.

On the one hand, $\hat{c}^{1}(\mathbf{x})$ and $\hat{c}^{1}(\mathbf{y})$ are both equal to $c_{r}(x_{r})+x_{r}c'_{r}(x_{r})$ for all $r\in \mathcal{R}^{\sharp}$. On the other hand, for all $r\in \mathcal{R}^{\flat}$, as $\hat{c}^{1}(\mathbf{x})=c_{r}(x_{r})+x^{1}_{r}c'_{r}(x_{r})$ and $\hat{c}^{1}(\mathbf{y})=c_{r}(x_{r})+y^{1}_{r}c'_{r}(x_{r})$, it follows from $\hat{c}^{1}(\mathbf{x})=\hat{c}^{1}(\mathbf{y})$ that $x^{1}_{r}=y^{1}_{r}$. Therefore, $T=\sum_{r\in \mathcal{R}^{\flat}}x^{1}_{r}+\sum_{r\in \mathcal{R}^{\sharp}}x_{r}=\sum_{r\in \mathcal{R}^{\flat}}y^{1}_{r}+\sum_{r\in \mathcal{R}^{\sharp}}y_{r}=S$, a contradiction. Hence, $\mathcal{R}_{+} \neq \emptyset$, and there exists some $r\in \mathcal{R}_{-}$ such that $y_{r}<x_{r}$.

Now, we will show that $Y^{0}(T)<Y^{0}(S)$, $Y^{1}(T)<Y^{1}(S)$ and $Y(T)<Y(S)$ in eight steps.

(a) Let us prove that there exists some $ s \in \mathcal{R}_{+}$ such that $y^{0}_{s}>0$.

If for all $s \in\mathcal{R}_{+}$, $y^{0}_{s}=0$, then $y^{1}_{s}=y_{s}>x_{s}\geq x^{1}_{s}$ and, consequently, $\hat{c}^{1}(\mathbf{y})= c_{s}(y_{s})+y^{1}_{s}c'(y_{s})>c_{r}(x_{s})+x^{1}_{s}c'(x_{s})=\hat{c}^{1}(\mathbf{x})$. Moreover, $\sum_{s\in\mathcal{R}_{+}}y^{1}_{s} > \sum_{s\in\mathcal{R}_{+}}x^{1}_{s}$. But $\sum_{r\in\mathcal{R}}y^{1}_{r}=S < T = \sum_{r\in\mathcal{R}}x^{1}_{r}$. Therefore, $\sum_{t\in\mathcal{R}_{-}}y^{1}_{t} < \sum_{t\in\mathcal{R}_{-}}x^{1}_{t}$. In particular, there exists some $r\in \mathcal{R}_{-}$ such that $y^{1}_{r} < x^{1}_{r}$. Since $y_{r}\leq x_{r}$, $\hat{c}^{1}(\mathbf{y})\leq c_{r}(y_{r})+y^{1}_{r}c'(y_{r})<c_{r}(x_{r})+x^{1}_{r}c'(x_{r})=\hat{c}^{1}(\mathbf{x})$, a contradiction.

(b) Let us show that $c^{0}(\mathbf{y})>c^{0}(\mathbf{x})$.

Choose the previous $s\in \mathcal{R}_{+}$ with $y^{0}_{s}>0$, and recall that $y_{s}>x_{s}$. Then, $Y^{0}(S)=c^{0}(\mathbf{y})=c_{s}(y_{s})>c_{s}(x_{s})\geq c^{0}(\mathbf{x})=Y^{0}(T)$. One deduces that $Y^{0}(T)$ is strictly decreasing in $T$ on $[\,\tilde{T},\, 1\,]$.

(c) For all $r\in \mathcal{R}_{-}$, $x^{0}_{r}=0$, because $y_{r}\leq x_{r}$ and, consequently, $c_{r}(x_{r}) \geq c_{r}(y_{r}) \geq c^{0}(\mathbf{y}) > c^{0}(\mathbf{x})$.

(d) Let us show that $\hat{c}^{1}(\mathbf{y})<\hat{c}^{1}(\mathbf{x})$.

Recall that there exists $r\in \mathcal{R}_{-}$ such that $y_{r}<x_{r}$. Then, $y^{1}_{r}\leq y_{r} < x_{r}$, and $x^{1}_{r}=x_{r}$ according to (c). Therefore, $\hat{c}^{1}(\mathbf{y})\leq c_{r}(y_{r})+y^{1}_{r}c'(y_{r})<c_{r}(x_{r})+x^{1}_{r}c'(x_{r})=\hat{c}^{1}(\mathbf{x})$.

(e) One can show that, for all $s\in \mathcal{R}_{+}$, $y^{1}_{s}<x^{1}_{s}$ and, consequently, $y^{0}_{s}>0$, $x^{1}_{s}>0$.

Indeed, for all $s\in \mathcal{R}_{+}$, $y_{s}>x_{s}\geq 0$, hence $y^{1}_{s}>0$. If there exists some $s\in \mathcal{R}_{+}$ such that $y^{1}_{s}\geq x^{1}_{s}$, then $\hat{c}^{1}(\mathbf{y})=c_{s}(y_{s})+y^{1}_{s}c'(y_{s}) > c_{s}(x_{s})+x^{1}_{s}c'(x_{s}) \geq \hat{c}^{1}(\mathbf{x})$, i.e. $\hat{c}^{1}(\mathbf{y})>\hat{c}^{1}(\mathbf{x})$. This contradicts (d).

It follows from the fact that $y_{s}>x_{s}$ that $y^{0}_{s} > x^{0}_{s}\geq 0$.  Besides, $x^{1}_{s}>y^{1}_{s}\geq 0$.

(f) For all $r\in \mathcal{R}_{-}$ and $s\in \mathcal{R}_{+}$, $c_{r}(x_{r}) > c_{s}(x_{s})$, because $c_{r}(x_{r})\geq c_{r}(y_{r}) \geq c^{0}(\mathbf{y}) = c_{s}(y_{s}) > c_{s}(x_{s})$.

(g) Let us define an auxiliary flow $\mathbf{z}$ in the game $\Gamma(\mathcal{R},\mathbf{c},(1-T;\,T))$ by
\begin{align*}
\begin{cases}
\, z^{1}_{s}=y_{s}-x^{0}_{s},\;z^{0}_{s}=x^{0}_{s},\quad\;\, & s\in  \mathcal{R}_{+};\\
\, z^{1}_{r}=y_{r},\;z^{0}_{r}=0, & r\in  \mathcal{R}_{-}.
\end{cases}
\end{align*}
Clearly, $\mathbf{z}'=\mathbf{y}'$, i.e. for all $r\in \mathcal{R}$, $z_{r}=y_{r}$, and
\begin{align*}
&  x^{1}_{s}< z^{1}_{s} \leq y_{s},\quad z^{0}_{s}=x^{0}_{s},\hspace{2cm} s\in \mathcal{R}_{+},\\
&  z^{1}_{r}=z_{r}\leq x_{r}=x^{1}_{r},\quad z^{0}_{r}=x^{0}_{r},\hspace{1.2cm} r\in \mathcal{R}_{-}.
\end{align*}

Now, we are ready to prove that the total cost to the coalition of weight $T$ at $\mathbf{z}$ is higher than that at $\mathbf{x}$, i.e. $u^{1}_{T}(\tilde{\mathbf{y}})>u^{1}_{T}(\mathbf{x})$ (the subscript $T$ is added to stress the weight of the coalition in question). Indeed, since for all $s\in \mathcal{R}_{+}$ and for all $r \in \mathcal{R}_{-}$ such that $x_{r}>0$,
\begin{equation}\label{unique_marginal}
\hat{c}^{1}(\mathbf{x}) = c_{s}(x_{s})+x^{1}_{s}c'(x_{s})=c_{r}(x^{1}_{r})+x^{1}_{r}c'(x^{1}_{r}),
\end{equation}
one deduces that, for all $s\in \mathcal{R}_{+}$, $r \in \mathcal{R}_{-}$ such that $x_{r}>0$, and for all $x > x^{1}_{s}$ and $y$ such that $0\leq y \leq x^{1}_{r}$,
\begin{equation}\label{eq:xx}
 c_{s}(x+x^{0}_{s})+x\,c'_{s}(x+x^{0}_{s}) > \hat{c}^{1}(\mathbf{x}) \geq c_{r}(y)+y\,c'_{r}(y).
\end{equation}
Then
\begin{align*}
u^{1}_{T}(\mathbf{z}) - u^{1}_{T}(\mathbf{x}) & = \sum_{s\in \mathcal{R}_{+}}\bigl[\, z^{1}_{s}c_{s}(z_{s})-x^{1}_{s}c_{s}(x_{s})\,\bigr]-\sum_{r\in \mathcal{R}_{-}}\bigl[\, x^{1}_{r}c_{r}(x^{1}_{r})-z^{1}_{r}c_{r}(z^{1}_{r})\,\bigr]\\
& = \sum_{s\in \mathcal{R}_{+}} \int^{z^{1}_{s}}_{x^{1}_{s}}\frac{\partial}{\partial x}\bigl[x\,c_{s}(x+x^{0}_{s})\bigr]\text{d}x - \sum_{r\in \mathcal{R}_{-}} \int^{x^{1}_{r}}_{z^{1}_{r}}\frac{\partial}{\partial x}\bigl[x\,c_{r}(x)\bigr]\text{d}x \\
& = \sum_{s\in \mathcal{R}_{+}} \int^{z^{1}_{s}}_{x^{1}_{s}} \bigl[ c_{s}(x+x^{0}_{s})+x\,c'_{s}(x+x^{0}_{s})\bigr]\text{d}x - \sum_{r\in \mathcal{R}_{-}} \int^{x^{1}_{r}}_{z^{1}_{r}} \bigl[ c_{r}(x)+x\,c'_{r}(x)\bigr]\text{d}x \\
& > \sum_{s\in \mathcal{R}_{+}} (z^{1}_{s}-x^{1}_{s})\,\hat{c}^{1}(\mathbf{x})-\sum_{r\in \mathcal{R}_{-}}(x^{1}_{r}-z^{1}_{r})\,\hat{c}^{1}(\mathbf{x}) = \sum_{r\in \mathcal{R}}(z^{1}_{r}-x^{1}_{r})\,\hat{c}^{1}(\mathbf{x})=(T-T)\,\hat{c}^{1}(\mathbf{x})=0
\end{align*}
The inequality above is due to \eqref{eq:xx}.

Next, notice the following three facts.

1) For all $r\in \mathcal{R}$, $z_{r}=y_{r}$ by the definition of $\mathbf{z}$,

2) For all $s\in \mathcal{R}_{+}$, $c_{s}(y_{s})=c^{0}(\mathbf{y})$ by (e), and

3) For all $r\in \mathcal{R}_{-}$, either $z^{1}_{r}=y_{r}=y^{1}_{r}$ or $z^{1}_{r}=y_{r}>y^{1}_{r}$. In the second case, $y^{0}_{r}>0$, which implies that $c_{r}(y_{r})=c^{0}(\mathbf{y})$.

These facts induce the relation between the total cost to the coalition of weight $T$ at $\mathbf{z}$ and the total cost to the coalition of weight $S$ at $\mathbf{y}$.
\begin{align*}
u^{1}_{T}(\mathbf{z}) & =\sum_{s\in \mathcal{R}_{+}}z^{1}_{s}\,c_{s}(z_{s}) + \sum_{r\in \mathcal{R}_{-}}z^{1}_{r}\,c_{r}(z_{r})\\
& = \sum_{s\in \mathcal{R}_{+}}\bigl[\,y^{1}_{s}+(z^{1}_{s}-y^{1}_{s})\,\bigr]c_{s}(z_{s}) + \sum_{r\in \mathcal{R}_{-}}\bigl[\,y^{1}_{r}+(z^{1}_{r}-y^{1}_{r})\,\bigr]c_{r}(z_{r})\\
& = \sum_{s\in \mathcal{R}_{+}} y^{1}_{s}\,c_{s}(y_{s}) + \sum_{r\in \mathcal{R}_{-}}y^{1}_{r}\,c_{r}(y_{r})+ \sum_{s\in \mathcal{R}_{+}}(z^{1}_{s}-y^{1}_{s})\,c_{s}(y_{s})+\sum_{r\in \mathcal{R}_{-}}(z^{1}_{r}-y^{1}_{r})\,c_{r}(y_{r})\\
& = u^{1}_{S}(\mathbf{y})+\sum_{r\in \mathcal{R}}(z^{1}_{r}-y^{1}_{r})\,c^{0}(\mathbf{y})\\
& = u^{1}_{S}(\mathbf{y}) + (T-S)\,c^{0}(\mathbf{y}) = S\cdot Y^{1}(S) + (T-S)\,c^{0}(\mathbf{y}).
\end{align*}

Recall that $u^{1}_{T}(\mathbf{z})>u^{1}_{T}(\mathbf{x})$. Then,
\begin{equation}\label{eq:comp}
T \cdot Y^{1}(T)= u^{1}_{T}(\mathbf{x})\,<\, u^{1}_{T}(\mathbf{z})=S\cdot Y^{1}(S) + (T-S)\,c^{0}(\mathbf{y}).
\end{equation}

This implies that $ Y^{1}(S) > Y^{1}(T)$. Because, otherwise, according to (ii), $Y^{0}(S)=c^{0}(\mathbf{y})< Y^{1}(S)$. Then,
\begin{equation*}
 S\cdot Y^{1}(S) + (T-S)\,c^{0}(\mathbf{y}) < T\cdot Y^{1}(S) \leq T \cdot Y^{1}(T),
\end{equation*}
which contradicts \eqref{eq:comp}.

Therefore, $Y^{1}(S) > Y^{1}(T)$. One deduces that $Y^{1}(T)$ is strictly decreasing in $T$ on $[\,\tilde{T},\, 1\,]$.

(h) Finally, let us prove that $Y(S) > Y(T)$, i.e. the social cost at $\mathbf{y}$ is higher than at $\mathbf{x}$. In other words, $Y(T)$ is strictly decreasing in $T$ on $[\,\tilde{T},\, 1\,]$.

Indeed, \eqref{unique_marginal} implies that, for all $s\in \mathcal{R}_{+}$ and $r \in \mathcal{R}_{-}$ such that $x_{r}>0$,
\begin{equation*}
c_{s}(x_{s})+x_{s}c'(x_{s})\geq \hat{c}^{1}(\mathbf{x}) \geq c_{r}(x_{r})+x_{r}c'(x_{r}).
\end{equation*}
Then, for all $s\in \mathcal{R}_{+}$ and $r \in \mathcal{R}_{-}$ such that $x_{r}>0$, for all $u>x_{s}$ and $v$ such that $0\leq v \leq x_{r}$,
\begin{equation}\label{eq:social}
 c_{s}(u)+u\,c'_{s}(u) > \hat{c}^{1}(\mathbf{x}) \geq c_{r}(v)+v\,c'_{r}(v).
\end{equation}
Thus,
\begin{align*}
Y(S)-Y(T) & =\sum_{s\in \mathcal{R}_{+}}\bigl[ \, y_{s}c_{s}(y_{s})-x_{s}c_{s}(x_{s}) \, \bigr] - \sum_{r\in \mathcal{R}_{-}}\bigl[ \, x_{r}c_{r}(x_{r})- y_{r}c_{r}(y_{r}) \, \bigr]\\
& = \sum_{s\in \mathcal{R}_{+}}\int^{y_{s}}_{x_{s}}\frac{\partial}{\partial u} u\,c_{s}(u)\,\text{d}u - \sum_{r\in \mathcal{R}_{-}}\int^{x_{r}}_{y_{r}}\frac{\partial}{\partial v} v\,c_{r}(v)\,\text{d}v\\
& =\int^{y_{s}}_{x_{s}} \bigl[ \,c_{s}(u)+u\,c'_{s}(u) \, \bigr]\,\text{d}u - \sum_{r\in \mathcal{R}_{-}}\int^{x_{r}}_{y_{r}}  \bigl[ \,c_{r}(v)+v\,c'_{r}(v) \, \bigr]\,\text{d}v\\
& > \sum_{s\in \mathcal{R}_{+}}(y_{s}-x_{s})\,\hat{c}^{1}(\mathbf{x})- \sum_{r\in \mathcal{R}_{-}}(x_{r}-y_{r})\,\hat{c}^{1}(\mathbf{x})=0,
\end{align*}
where the inequality is due to \eqref{eq:social}.
\end{proof}

\subsection{Individuals' cost and the composition of the players}
The previous results can be partially extended to the multiple coalitions case. Consider the following two composite games.
\begin{align*}
&\Gamma_{0}=\Gamma(\mathcal{R},\mathbf{c},\mathbf{T}),\;\,\;\mathbf{T} =\bigl(T^{0};T^{1},\ldots,T^{K}\bigr),\\
&\Gamma_{1}=\Gamma(\mathcal{R},\mathbf{c},\mathbf{T}'),\;\;\mathbf{T}'=\bigl(T^{0}+\delta T;T^{1},\ldots,T^{l-1},T^{l}-\delta T, T^{l+1},\ldots,T^{K}\bigr),
\end{align*}
with $K\geq 1$, $1 \leq l \leq K$ and $0< \delta T < T^{l}$. Profile $\mathbf{T}'$ can be seen as obtained by $\mathbf{T}$ after the withdrawal from coalition $l$ of a group of members of total weight $\delta T$ who become individuals. Let $\mathbf{x}$ and $\mathbf{y}$ be, respectively, the CE of the game $\Gamma_{0}$ and that of the game $\Gamma_{1}$. The other notations are as before.

The following theorem shows that the individuals' cost is (weakly) higher at the CE in the game $\Gamma_{1}$ than in the game $\Gamma_{0}$.
\begin{theorem}\label{thm:monotonie}
$c^{0}(\mathbf{x})\leq c^{0}(\mathbf{y})$.
\end{theorem}
\begin{proof}
Since the case $K=1$ is proven in Theorem~\ref{efficiency_size}, only the case $K\geq 2$ is treated here.

First, define two sets of arcs $\mathcal{R}_{-}=\{r\in \mathcal{R}\,|\,y_{r}\leq x_{r}\}$ and $\mathcal{R}_{+}=\mathcal{R}\setminus \mathcal{R}_{-}=\{r\in \mathcal{R}\,|\,y_{r}>x_{r}\}$.

If $\mathcal{R}_{+}=\emptyset$, then $\mathbf{x}'=\mathbf{y}'$ and $c^{0}(\mathbf{x})=c^{0}(\mathbf{y})$.

If $\mathcal{R}_{+} \neq \emptyset$, let us first prove that, for all $k\in \mathcal{K}\setminus\{l\}$, $\sum_{r\in \mathcal{R}_{-}} y^{k}_{r} \geq \sum_{r\in \mathcal{R}_{-}} x^{k}_{r}$.

Suppose that $\sum_{r\in \mathcal{R}_{-}} y^{k}_{r} \leq \sum_{r\in \mathcal{R}_{-}} x^{k}_{r}$ and, consequently, $\sum_{s\in \mathcal{R}_{+}} y^{k}_{s} \geq \sum_{s\in \mathcal{R}_{+}} x^{k}_{s}$. Therefore, there is an arc $r\in \mathcal{R}_{-}$ and an arc $s\in \mathcal{R}_{+}$ such that $y^{k}_{r} \leq x^{k}_{r}$ and $y^{k}_{s} \geq x^{k}_{s}$.

For all such $r$ and $s$, if $x^{k}_{r}>0$ and $y^{k}_{s}>0$, then $\hat{c}^{k}(\mathbf{y}) \leq c_{r}(y_{r})+y^{k}_{r}c'_{r}(y_{r})
\leq c_{r}(x_{r})+x^{k}_{r}c'_{r}(x_{r}) = \hat{c}^{k}(\mathbf{x})
\leq\, c_{s}(x_{s})+x^{k}_{s}c'_{s}(x_{s}) < c_{s}(y_{s})+y^{k}_{s}c'_{s}(y_{s}) = \hat{c}^{k}(\mathbf{y})$,
a contradiction. In consequence, either $y^{k}_{r} = x^{k}_{r} =0$ or $y^{k}_{s} = x^{k}_{s} =0$. For this to be true, there can be two cases.

{\sc Case 1.} For all $s\in \mathcal{R}_{+}$ such that $y^{k}_{s} \geq x^{k}_{s}$, $y^{k}_{s} = x^{k}_{s} =0$. Then, there is no $s\in \mathcal{R}_{+}$ such that $y^{k}_{s} < x^{k}_{s}$ because, otherwise, $\sum_{s\in \mathcal{R}_{+}} y^{k}_{s} < \sum_{s\in \mathcal{R}_{+}} x^{k}_{s}$, which contradicts the hypothesis that $\sum_{r\in \mathcal{R}_{-}} y^{k}_{r} \leq \sum_{r\in \mathcal{R}_{-}} x^{k}_{r}$. Thus, for all $s\in \mathcal{R}_{+}$, $y^{k}_{s} = x^{k}_{s}=0$ and, consequently, $\sum_{r\in \mathcal{R}_{-}} y^{k}_{r} = \sum_{r\in \mathcal{R}_{-}} x^{k}_{r}= T^{k}$.

{\sc Case 2.} There exists some $s\in \mathcal{R}_{+}$ such that $y^{k}_{s} \geq x^{k}_{s}$ and $y^{k}_{s}>0$. Then, for all $r\in \mathcal{R}_{-}$ such that $y^{k}_{r} \leq x^{k}_{r}$, $y^{k}_{r} = x^{k}_{r} =0$. Therefore, there is no $r\in \mathcal{R}_{-}$ such that $y^{k}_{r} > x^{k}_{r}$ because, otherwise, $\sum_{r\in \mathcal{R}_{-}} y^{k}_{r} > \sum_{r\in \mathcal{R}_{-}} x^{k}_{r}$, which again contradicts the hypothesis. Thus, for all $r\in \mathcal{R}_{-}$, $y^{k}_{r} = x^{k}_{r} =0$ and, in consequence, $\sum_{r\in \mathcal{R}_{-}} y^{k}_{r} = \sum_{r\in \mathcal{R}_{-}} x^{k}_{r} =0$.

Hence, $\sum_{r\in \mathcal{R}_{-}} y^{k}_{r} \geq \sum_{r\in \mathcal{R}_{-}} x^{k}_{r}$ and, consequently, $\sum_{s\in \mathcal{R}_{+}} y^{k}_{s} \leq \sum_{s\in \mathcal{R}_{+}} x^{k}_{s}$. Besides, the equalities hold if, and only if, $\sum_{r\in \mathcal{R}_{-}} y^{k}_{r} = \sum_{r\in \mathcal{R}_{-}} x^{k}_{r}= T^{k}$ or $0$.

Since $\sum_{s\in \mathcal{R}_{+}} y_{s} > \sum_{s\in \mathcal{R}_{+}} x_{s}$ and $\sum_{s\in \mathcal{R}_{+}} y^{k}_{s} \leq \sum_{s\in \mathcal{R}_{+}} x^{k}_{s}$ for all $k\in \mathcal{K}\setminus \{l\}$, one deduces that $\sum_{s\in \mathcal{R}_{+}} \bigl(y^{0}_{s}+y^{l}_{s}\bigr) > \sum_{s\in \mathcal{R}_{+}} \bigl(x^{0}_{s}+x^{l}_{s}\bigr)\geq 0$.

Let us show that there exists some $t\in \mathcal{R}_{+}$ such that $y^{0}_{t}>0$. Indeed, if for all $s \in\mathcal{R}_{+}$, $y^{0}_{s}=0$, then $\sum_{s\in\mathcal{R}_{+}}y^{l}_{s} > \sum_{s\in\mathcal{R}_{+}}(x^{l}_{s}+x^{0}_{s})\geq \sum_{s\in\mathcal{R}_{+}}x^{l}_{s}$. Besides, $\sum_{r\in\mathcal{R}}y^{l}_{r}=T-\delta T < T = \sum_{r\in\mathcal{R}}x^{l}_{r}$, hence $\sum_{t\in\mathcal{R}_{-}}y^{l}_{t} < \sum_{t\in\mathcal{R}_{-}}x^{l}_{t}$. In particular, there exists $r\in \mathcal{R}_{-}$ such that $y^{l}_{r} < x^{l}_{r}$ and $s\in \mathcal{R}_{+}$ such that $y^{l}_{s} > x^{l}_{s}$. Then, $\hat{c}^{l}(\mathbf{y})\leq c_{r}(y_{r})+y^{l}_{r}c'(y_{r})<c_{r}(x_{r})+x^{l}_{r}c'(x_{r})=\hat{c}^{l}(\mathbf{x})$, and $\hat{c}^{l}(\mathbf{y})=c_{s}(y_{s})+y^{l}_{s}c'(y_{s})>c_{s}(x_{s})+x^{l}_{s}c'(x_{s})\geq \hat{c}^{l}(\mathbf{x})$, a contradiction. Thus, there exists $t\in \mathcal{R}_{+}$ such that $y^{0}_{t} > 0$ and, consequently, $c^{0}(\mathbf{y})=c_{t}(y_{t})>c_{t}(x_{t}) \geq c^{0}(\mathbf{x})$.
\end{proof}

However, the average cost to the coalition of weight $T^{l}$ (called coalition $l$) in $\Gamma_{0}$ is not necessarily lower than that to the coalition of weight $T^{l}-\delta T$ (called coalition $l'$) in $\Gamma_{1}$, and the former is not necessarily lower than the average cost to the group composed of coalition $l'$ and a set of individuals of total weight $\delta T$ in $\Gamma_{1}$ (the group corresponding to coalition $l$ in $\Gamma_{0}$). In other words, the other two results in Theorem~\ref{efficiency_size} (iii) cannot be extended to the multiple coalitions case. Here is an example.

\begin{example}\label{ex:comp}\rm
Composite game $\Gamma$ takes place in the same network as in Example~\ref{ex:1}, where the cost functions of arcs $r_{1}$ and $r_{2}$ are, respectively, $c_{1}(x)=x+10$ and $c_{2}(x)=10x+1$. Coalition $1$ has weight $T$ with $T\in (0,\frac{1}{2}]$; coalition $2$ has weight $\frac{1}{2}$; the total weight of the individuals is $\frac{1}{2}-T$. Then, the average cost to coalition $1$ at the CE is $\frac{91}{11}$ for $T\in (0,\frac{1}{5}]$ and $\frac{1}{99}[919-4(25T+\frac{4}{T})]$ for $T\in [\frac{1}{5},\frac{1}{2}]$, which is constant in $T$ on $(0,\frac{1}{5}]$, strictly increasing on $[\frac{1}{5},\frac{2}{5}]$ and strictly decreasing on $[\frac{2}{5},\frac{1}{2}]$. Therefore, it is not always decreasing in the size of the coalition. The average cost to the group of total weight $\frac{1}{2}$ composed of coalition $1$ and all the individuals is $\frac{91}{11}$ for $T\in (0,\frac{1}{5}]$ and $\frac{1}{99}[400(T-\frac{11}{40})^{2}+816.75]$ for $T\in [\frac{1}{5},\frac{1}{2}]$, which is constant in $T$ on $(0,\frac{1}{5}]$, strictly decreasing in $(\frac{1}{5},\frac{11}{40}]$ and strictly increasing on $[\frac{11}{40},\frac{1}{2}]$. In consequence, it is not always decreasing in $T$.
\end{example}

\section{Asymptotic behavior of composite games}\label{sect6}
This subsection studies the asymptotic behavior of composite games, when some coalitions are fixed and the size of the others vanish.

\begin{definition}[Admissible sequence of composite games and its limit game]
A sequence of composite games $\{\Gamma_{n}\}_{n\in \mathbb{N}^{*}}$, with $\Gamma_{n}=\Gamma(\mathcal{R},\mathbf{c},\mathbf{T}_{n})$ and $\mathbf{T}_{n}=(T^{0}_{n};\,T^{1}_{n},\,T^{2}_{n},\ldots,\,T^{K_{n}}_{n})$, is called {\em admissible} if $\{\mathbf{T}_{n}\}_{n\in \mathbb{N}^{*}}$ satisfies the following conditions.
\begin{enumerate}
 \item There is a constant $L \in \mathbb{N}$, and $L$ strictly positive constants $\{T^{1},\,T^{2},\ldots,\,T^{L}\}$ such that $\sum^{L}_{k=1}T^{k}<1$. For all $n$, $K_{n}>L$, and $T^{i}_{n}=T^{i}$ for $i=1,\ldots,L$.
 \item $\delta_{n}=\max_{L<k\leq K_{n}}T^{k}_{n}$. And $\delta_{n}\rightarrow 0$ as $n\rightarrow \infty$.
\end{enumerate}

The $L$-coalition composite game $\Gamma_{0}=\Gamma(\mathcal{R},\mathbf{c},(\tilde{T}^{0};\,T^{1},\,T^{2},\ldots,\,T^{L}))$ is called the {\em limit game} of the sequence $\{\Gamma_{n}\}_{n\in \mathbb{N}^{*}}$, where $\tilde{T}^{0}=1-\sum^{L}_{k=1}T^{k}$.
\end{definition}
\begin{remark}\rm
Condition (i) means that there are $L$ coalitions fixed all along the sequence $\{\Gamma_{n}\}_{n\in \mathbb{N}^{*}}$, and the total weight of the remaining coalitions and the individuals is fixed to $\tilde{T}^{0}$. Condition (ii) means that the other coalitions are vanishing along the sequence and, necessarily, $K_{n}$ tends to infinity. 
\end{remark}

\noindent {\sc Notations} As before, in the game $\Gamma_{n}$, $\mathbf{x}^{*}_{n}=(\mathbf{x}^{*k}_{n})^{K_{n}}_{k=0}$ is the CE, where $\mathbf{x}^{*0}_{n}$ is the flow of the individuals, and $\mathbf{x}^{*k}_{n}$ the flow of coalition $k$. Besides, $Y^{0}(\mathbf{x}^{*}_{n})$ is the individuals' cost and $Y^{k}(\mathbf{x}^{*}_{n})$ the average cost to coalition $k$ at CE.

The aggregate flows are defined as ${\mathbf{x}^{*}_{n}}'=(\mathbf{y}^{*}_{n},\,\mathbf{x}^{*1}_{n},\,\mathbf{x}^{*2}_{n},\,\ldots,\,\mathbf{x}^{*L}_{n})$, where $\mathbf{y}^{*}_{n}=(y^{*}_{n,r})_{r\in \mathcal{R}}$, $y^{*}_{n,r}=x^{*0}_{n,r}+\sum^{K_{n}}_{k=L+1}x^{*k}_{n,r}$. Thus, $\mathbf{y}^{*}_{n}$ is the aggregate flow of the individuals in addition to all the coalitions different from the $L$ fixed ones. Notice that this is different from the definition of aggregate flow in the previous sections.

$F_{n}=\{\mathbf{x}\in \mathbb{R}^{R\times(1+K_{n})}\,|\, \mathbf{x}\geq \mathbf{0};\;\forall\, k=0 \text{ or } k\in \mathcal{K}, \,\sum_{r\in \mathcal{R}}x^{k}_{r}=T^{k}_{n}\}$ is the feasible flow set. $F_{0}=\{\mathbf{x}\in \mathbb{R}^{R\times(1+L)}\,|\, \mathbf{x}\geq \mathbf{0};\;\forall\, k\in \mathcal{K},\,\sum_{r\in \mathcal{R}}x^{k}_{r}=T^{k}; \;\sum_{r\in \mathcal{R}}x^{0}_{r}=\tilde{T}^{0}\}$ is the feasible aggregate flow set. Notice that it is common to all the games in $\{\Gamma_{n}\}_{n\in \mathbb{N}^{*}}$.

In $\Gamma_{0}$, $\mathbf{x}^{*}=(\mathbf{y}^{*},\,\mathbf{x}^{*1},\,\ldots,\,\mathbf{x}^{*L})$ is the CE, where $\mathbf{y}^{*}=\mathbf{x}^{*0}$ is the flow of the individuals, and $\mathbf{x}^{*k}$ the flow of coalition $k$. $Y^{0}(\mathbf{x}^{*})$ is the individuals' cost and $Y^{k}(\mathbf{x}^{*})$ the average cost to coalition $k$ at CE. The feasible flow set is $F_{0}$.\\

The following theorem states that the CE of $\Gamma_{n}$ converges to the CE of $\Gamma_{0}$. Hence, it justifies the name `limit game'.
\begin{theorem}[Convergence of admissible composite games]\label{limit_mixed}
Suppose that $\{\Gamma_{n}\}_{n\in \mathbb{N}^{*}}$ is a sequence of admissible games satisfying Assumption~\ref{cost_assumption_1}. Let $\Gamma_{0}$ be its limit game. Then, ${\mathbf{x}^{*}_{n}}' \rightarrow \mathbf{x}^{*}$ as $n\rightarrow \infty$. In particular, $Y^{k}(\mathbf{x}^{*}_{n}) \rightarrow Y^{k}(\mathbf{x}^{*})$ for $k=1,\ldots,L$, and $Y^{k}(\mathbf{x}^{*}_{n}) \rightarrow Y^{0}(\mathbf{x}^{*})$ for $k=0$ and $k>L$.
\end{theorem}
\begin{proof}
Let us begin by writing the variational inequality condition for the CE's $\mathbf{x}^{*}_{n}$ and $\mathbf{x}^{*}$.

By Proposition~\ref{eq_characterize}, $\mathbf{x}^{*}_{n}$ is the CE of $\Gamma_{n}$ if, and only if,
\begin{equation}\label{VI_n}
\left\langle\,\mathbf{c}(\mathbf{x}^{*}_{n}),\,\mathbf{x}^{0}_{n}-\mathbf{x}^{*0}_{n}\,\right\rangle+\sum^{K_{n}}_{k=1}\left\langle\,\hat{\mathbf{c}}^{k}(\mathbf{x}^{*}_{n}),\,\mathbf{x}^{k}_{n}-\mathbf{x}^{*k}_{n}\,\right\rangle\geq 0, \quad \forall\; \mathbf{x}_{n} \in F_{n},
\end{equation}
and $\mathbf{x}^{*}$ is the CE of $\Gamma_{0}$ if, and only if,
\begin{equation}\label{VI_limit}
\left\langle\,\mathbf{c}(\mathbf{x}^{*}),\,\mathbf{y}-\mathbf{y}^{*}\,\right\rangle+\sum^{L}_{k=1}\left\langle\,\hat{\mathbf{c}}^{k}(\mathbf{x}^{*}),\,\mathbf{x}^{k}-\mathbf{x}^{*k}\,\right\rangle\geq 0, \quad \forall\; \mathbf{x}= (\mathbf{y},\,\mathbf{x}^{1},\,\ldots,\,\mathbf{x}^{L}) \in F_{0}.
\end{equation}

Due to Assumption~\ref{cost_assumption_1}, one can find a constant $M$ such that $M>\sup_{r\in \mathcal{R},\,x\in [0,\,1]}\{|c'_{r}(x)|\}$. Set $\epsilon_{n}=2\delta_{n}MR$ so that $\epsilon_{n}$ tends to $0$. Let us show that, for all $n$, the aggregate flow ${\mathbf{x}^{*}_{n}}'$ in $F_{0}$ satisfies
\begin{equation}\label{converge}
\left\langle\,\mathbf{c}({\mathbf{x}^{*}_{n}}'),\,\mathbf{y}-\mathbf{y}^{*}_{n}\,\right\rangle+\sum^{L}_{k=1}\left\langle\,\hat{\mathbf{c}}^{k}({\mathbf{x}^{*}_{n}}'),\,\mathbf{x}^{k}-\mathbf{x}^{*k}_{n}\,\right\rangle\geq -\epsilon_{n}, \quad \forall\; \mathbf{x}= (\mathbf{y},\,\mathbf{x}^{1},\,\ldots,\,\mathbf{x}^{L})\in F_{0}.
\end{equation}

Indeed, for any $\mathbf{x}= (\mathbf{y},\,\mathbf{x}^{1},\,\ldots,\,\mathbf{x}^{L}) \in F_{0}$, one can find $\mathbf{x}_{n}=(\mathbf{x}^{k}_{n})^{K_{n}}_{k=0} \in F_{n}$ such that
\begin{equation}\label{induce}
\begin{cases}
\,\mathbf{x}^{k}_{n}=\mathbf{x}^{k}, \quad &k=1,\ldots, L;\\
\,x^{0}_{n,r}+\sum^{K_{n}}_{k=L+1}x^{k}_{n,r}=y_{n,r}, \quad &\forall\; r\in \mathcal{R}.
\end{cases}
\end{equation}
For example, take $x^{0}_{n,r}=y_{n,r}T^{0}/\tilde{T}^{0}$, $x^{k}_{n,r}=y_{n,r}T^{k}/\tilde{T}^{0}$ for $k=L+1,\,\ldots,\,K_{n}$. Then, by \eqref{VI_n},
\begin{align*}
& \bigl\langle\,\mathbf{c}(\mathbf{x}^{*}_{n}),\,\mathbf{x}^{0}_{n}-\mathbf{x}^{*0}_{n}\,\bigr\rangle+\sum^{K_{n}}_{k=1}\bigl\langle\,\hat{\mathbf{c}}^{k}(\mathbf{x}^{*}_{n}),\,\mathbf{x}^{k}_{n}-\mathbf{x}^{*k}_{n}\,\bigr\rangle\geq 0 \\
\Rightarrow \; & \bigl\langle\,\mathbf{c}(\mathbf{x}^{*}_{n}),\,\mathbf{x}^{0}_{n}-\mathbf{x}^{*0}_{n}\,\bigr\rangle+\sum^{L}_{k=1}\bigl\langle\,\hat{\mathbf{c}}^{k}(\mathbf{x}^{*}_{n}),\,\mathbf{x}^{k}_{n}-\mathbf{x}^{*k}_{n}\,\bigr\rangle+\sum^{K_{n}}_{k=L+1}\bigl\langle\,\mathbf{c}(\mathbf{x}^{*}_{n})+\mathbf{x}^{*k}_{n}\dot{\mathbf{c}}(\mathbf{x}^{*}_{n}),\,\mathbf{x}^{k}_{n}-\mathbf{x}^{*k}_{n}\,\bigr\rangle\geq 0\\
\Rightarrow  \;& \Bigl\langle\,\mathbf{c}(\mathbf{x}^{*}_{n}),\,\mathbf{x}^{0}_{n}+\sum^{K_{n}}_{k=L+1}\mathbf{x}^{k}_{n}-\mathbf{x}^{*0}_{n}-\sum^{K_{n}}_{k=L+1}\mathbf{x}^{*k}_{n}\,\Bigr\rangle+\sum^{L}_{k=1}\bigl\langle\,\hat{\mathbf{c}}^{k}(\mathbf{x}^{*}_{n}),\,\mathbf{x}^{k}_{n}-\mathbf{x}^{*k}_{n}\,\bigr\rangle\\
\displaybreak[3]
& \geq -\sum^{K_{n}}_{k=L+1}\bigl\langle\,\mathbf{x}^{*k}_{n}\dot{\mathbf{c}}(\mathbf{x}^{*}_{n}),\,\mathbf{x}^{k}_{n}-\mathbf{x}^{*k}_{n}\,\bigr\rangle \geq -\sum^{K_{n}}_{k=L+1}\bigl\langle\,\delta_{n}M,\,\mathbf{x}^{k}_{n}-\mathbf{x}^{*k}_{n}\,\bigr\rangle\\
&= -\,\Bigl\langle\,\delta_{n}M,\,\sum^{K_{n}}_{k=L+1}\mathbf{x}^{k}_{n}-\sum^{K_{n}}_{k=L+1}\mathbf{x}^{*k}_{n}\,\Bigr\rangle \geq -2\delta_{n}MR.
\end{align*}

By \eqref{induce}, this is just $\langle\,\mathbf{c}({\mathbf{x}^{*}_{n}}'),\,\mathbf{y}-\mathbf{y}^{*}_{n}\,\rangle+\sum^{L}_{k=1}\bigl\langle\,\hat{\mathbf{c}}^{k}({\mathbf{x}^{*}_{n}}'),\,\mathbf{x}^{k}-\mathbf{x}^{*k}_{n}\,\bigr\rangle\geq -\epsilon_{n}$.

As $F_{0}$ is a compact subset of $\mathbb{R}^{R\times(1+L)}$, $\{{\mathbf{x}^{*}_{n}}'\}_{n\in \mathbb{N}^{*}}$ in $F_{0}$ admits accumulation points. For any convergent subsequence of $\{{\mathbf{x}^{*}_{n}}'\}_{n\in \mathbb{N}^{*}}$ (which is still denoted by $\{{\mathbf{x}^{*}_{n}}'\}_{n\in \mathbb{N}^{*}}$ for simplicity), let $\tilde{\mathbf{x}}= (\tilde{\mathbf{y}},\,\tilde{\mathbf{x}}^{1},\,\ldots,\,\tilde{\mathbf{x}}^{L})$ be its accumulation point. Let $n$ tend to infinity in \eqref{converge}. Then, by the continuity of the marginal cost functions and the fact that $\epsilon_{n}$ tends to 0,
\begin{equation*}
\left\langle\,\mathbf{c}(\tilde{\mathbf{x}}^{L}),\,\mathbf{y}-\tilde{\mathbf{y}}\,\right\rangle+\sum^{L}_{k=1}\left\langle\,\hat{\mathbf{c}}^{k}(\tilde{\mathbf{x}}),\,\mathbf{x}^{k}-\tilde{\mathbf{x}}^{k}\,\right\rangle\geq 0,\quad \forall\; \mathbf{x}= (\,\mathbf{y},\,\mathbf{x}^{1},\,\ldots,\,\mathbf{x}^{L}\,) \in F_{0}.
\end{equation*}
According to \eqref{VI_limit}, this implies that $\tilde{\mathbf{x}}=\mathbf{x}^{*}$.

Therefore, ${\mathbf{x}^{*}_{n}}'$ converges to $\mathbf{x}^{*}$ as $n$ tends to infinity. This induces immediately that $Y^{k}(\mathbf{x}^{*}_{n})$ tends to $Y^{k}(\mathbf{x}^{*})$ for $k=1,\ldots,L$, and $Y^{k}(\mathbf{x}^{*}_{n})$ tends to $Y^{0}(\mathbf{x}^{*})$ for $k=0$ and $k>L$.
\end{proof}
\begin{remark}\rm
When $T^{0}=0$ and $L=0$, Theorem~\ref{limit_mixed} shows that the NE of an atomic splittable game with only coalitions and no individuals converges to the WE of the corresponding nonatomic game, when the coalitions split into smaller and smaller ones. This result is obtained by Haurie and Marcotte~\cite{Hau85}, but only for the case where the coalitions split into equal-size ones. Theorem~\ref{limit_mixed} is an extension of their result in three aspects. First, the coalitions do not have equal size. Second, the games are composite. Finally, some coalitions are fixed at a nonnegligible weight.
\end{remark}
\section{Some problems for future research}\label{sect7}
This section presents some directions for further studies.
\subsection{Backward induction}
Consider a two-stage extensive form game with an underlying network ($\mathcal{R},\,\mathbf{c}$) and a set of nonatomic individuals $[\,0,\,1]$. At the first stage, the individuals are given $K+1$ choices $\{s_{0},\,s_{1},\,\ldots,\,s_{K}\}$, where $K$ is a fixed number in $\mathbb{N}^{*}=\mathbb{N}\setminus \{0\}$. The players who choose $s_{0}$ are called individuals, and those who choose $s_{k}$ ($1\leq k \leq K$) are considered as members of coalition $k$. If there are $L$ coalitions having nonnegligible weights ($0\leq L\leq K$) then, at the second stage, the individuals and $L$ coalitions play the composite routing game $\Gamma(\mathcal{R},\mathbf{c},\mathbf{T})$. By Lemma~\ref{me_flow_1}, the players who choose $s_{0}$ at the first stage have the lowest cost at the end. Therefore, by a backward induction, the only subgame perfect equilibrium of this two-stage game consists in having all the players choosing $s_{0}$ at the first stage.

\subsection{Composition-decision games}
In a two-player composition-decision game, each player is atomic with a splittable flow. The weight of player $I$ is $T$, while that of player $II$ is $1-T$. Each player chooses a pair of representatives consisting of a coalition and a group of individuals, whose total weight is her own weight. The cost to each player is defined as the average equilibrium cost to her representatives in a composite routing game played by all the representatives. Consider two simple models where the game reduces to a one-player game.
\paragraph{Model 1: One atomic player faces individuals.}
Player $I$ has two strategies. Strategy 1 consists in choosing a coalition of weight $T$, while strategy 2 consists in choosing a group of individuals of weight $T$. Player $II$ always chooses a group of individuals of weight $1-T$.

If player $I$ chooses strategy 1, the costs to the two players are, respectively, the equilibrium cost to the unique coalition and that to the individuals in the composite game $\Gamma(\mathcal{R},\mathbf{c},(1-T;T))$, i.e. $Y^{1}(T)$ for player $I$, $Y^{0}(T)$ for player $II$. If player $I$ chooses strategy 2, both the costs to the two players are $W$, the equilibrium cost in the nonatomic game $\Gamma(\mathcal{R},\mathbf{c})$.

Theorem~\ref{efficiency_size}(ii) shows that, if $0<T\leq \tilde{T}$, then the two strategies make no difference to player $I$. If $\tilde{T}<T<1$, player $I$'s only best reply is strategy 1, and her cost is $Y^{1}(T)$, which is lower than $W$. Strategy 1 dominates strategy 2.
\paragraph{Model 2: One group faces a coalition.} Player $I$ has the same two strategies as in Model 1. Player $II$ always chooses a coalition of weight $1-T$.

If player $I$ chooses strategy 1, the costs to the two players are the equilibrium costs in the two-coalition game $\Gamma(\mathcal{R},\mathbf{c},(0;T,1-T))$. If player $I$ chooses strategy 2, the costs to the two players are, respectively, the equilibrium cost to the individuals and that to the unique coalition in the composite game $\Gamma(\mathcal{R},\mathbf{c},(T;1-T))$, i.e. $Y^{0}(1-T)$ for player $I$, $Y^{1}(1-T)$ for player $II$.

Does strategy 1 still dominate strategy 2? The answer is negative. In the second part of Example~\ref{ex:comp}, strategy 2 is a better response than strategy 1.

\section*{Acknowledgments}
I am most grateful to Professor Sylvain Sorin for all his advice, help and patience throughout the writing of this work. I thank the editors and the two anonymous referees of {\em Mathematics of Operations Research} for their helpful suggestions and remarks.

\end{document}